\documentclass[12pt]{amsart}

\usepackage{tabularx} 
\usepackage{graphicx}

\usepackage{xcolor}
\usepackage[margin = 1in]{geometry} 
\usepackage[colorlinks=true, linkcolor=blue, citecolor=blue]{hyperref}

\usepackage{caption}
\usepackage{subfig}

\newtheorem{theorem}{Theorem}[section]
\newtheorem{lemma}{Lemma}[section]
\newtheorem{corollary}{Corollary}[section]
\newtheorem{proposition}{Proposition}[section]
\newtheorem{Remark}{Remark}[section]
\numberwithin{equation}{section}
\theoremstyle{definition}

\newtheorem{defn}{Definition}[section]

\begin{document}
\title[Bivariate Gamma Generalized Laplace Law]{Modeling Stock Returns and Volatility Using Bivariate Gamma Generalized Laplace Law} 

\author{Tomasz J. Kozubowski, Andrey Sarantsev, James A. Spiker}

\begin{abstract}  
We consider a generalization of the variance-gamma (generalized asymmetric Laplace) distribution, defined as a normal mean - variance mixture with a gamma mixing distribution. While this model is typically studied in the univariate setting, we assume that the gamma mixing variable is observed alongside the primary variable, resulting in a bivariate framework. In this setting, maximum likelihood estimation becomes significantly simpler than in the standard univariate case, reducing to a form of classical linear regression. We derive explicit expressions for the resulting estimators. For certain parameter configurations, the estimators exhibit nonstandard convergence rates, exceeding the usual square-root rate. Finally, we illustrate the applicability of this model in financial contexts by analyzing stock index returns and associated volatility for several major indices.
\end{abstract}



\date{\today. \textit{Address:} University of Nevada, Reno. Department of Mathematics and Statistics. Corresponding author: AS. Email: \texttt{asarantsev@unr.edu}}

\maketitle
\thispagestyle{empty}
\section{Introduction}
\label{BGGL.intro.sec}

The Laplace distribution with location parameter $\delta \in \mathbb{R}$ and variance $\sigma^2 > 0$, given by the probability density function (PDF)
\begin{equation}
\label{lap.pdf}
f(x) = \frac{1}{\sqrt{2}\sigma} \exp \left( -\frac{\sqrt{2}|x - \delta|}{\sigma} \right), \quad x \in \mathbb{R},
\end{equation}
along with its numerous generalizations, has played an increasingly important role in probability and applied modeling over the past several decades (see, e.g., \cite{Book}). This distribution is a special case of a {\it mean--variance Gaussian mixture} model: a Laplace random variable $Y$ with the PDF in~\eqref{lap.pdf} admits the stochastic representation
\begin{equation}
\label{eq:VG}
Y = \delta + \mu X + \sigma \sqrt{X} Z,
\end{equation}
where $\mu = 0$, $Z$ is standard normal, and $X$ is standard exponential, independent of $Z$. Allowing $\mu \neq 0$ yields the more flexible Asymmetric Laplace (AL) family (see, e.g., \cite{Book}), which in turn generalizes naturally to the Generalized Asymmetric Laplace (GAL) distribution (see, e.g., \cite{Koz13}) when $X$ follows a gamma distribution. Since a GAL random variable can be viewed as Gaussian with a stochastic, gamma distributed variance (and mean), this model is also known as the {\it variance gamma} distribution, introduced in \cite{VG} in a financial context (see also the survey \cite{Survey}). Historically, however, its roots reach much earlier: as discussed in \cite{Koz13}, it appeared already in 1929 as the distribution of the empirical correlation coefficient based on bivariate Gaussian samples \cite{Pearson}.

The GAL family enjoys several attractive properties: Finite moments of all orders and a moment generating function finite in a neighborhood of zero (see, e.g., \cite{Survey}). The parameter $\mu$ controls skewness ($\mu = 0$ implies symmetry), while the shape parameter $\alpha$ of the mixing gamma distribution governs kurtosis. Together with location and scale, these four parameters allow the model to match arbitrary mean, variance, skewness, and kurtosis, offering much greater flexibility than the classical normal family. 

This flexibility makes the GAL distribution especially useful in finance, where asset returns often exhibit heavier than Gaussian tails \cite[Section 8.4]{Book}, \cite{Seneta}. A common interpretation models stock prices as geometric Brownian motion evolving in business time: A nonlinear, random time scale represented by a nondecreasing L\'evy subordinator that accelerates with market activity. This leads to a subordinated Brownian motion, or {\it Laplace motion}, whose increments follow an (asymmetric) Laplace law (see \cite{Clark} and \cite[Section 8.4]{Book}). While univariate mean-variance Gaussian mixtures with conditional distribution
\begin{equation}
\label{eq:mixture}
Y \mid X \sim \mathcal{N}(\delta + \mu X, \sigma^2 X),
\end{equation}
with $X$ distributed as Gamma, have been extensively studied (see, e.g., \cite{BKS1982}) and have found broad application in hierarchical modeling, Bayesian inference, stochastic processes, clustering, and machine learning, among others, the bivariate extension (where $X$ and $Y$ arise jointly as dependent random components) appears to have received relatively little attention.

In this work, we address this gap by extending the GAL model to a bivariate setting. Specifically, we introduce a joint distribution for $(X, Y)$ where $Y$ follows~\eqref{eq:VG} with $X$ gamma distributed and independent of $Z$. We refer to this new family as the  {\it Bivariate Gamma-Generalized Laplace} (BGGL) distribution. Compared with the univariate GAL model, the maximum likelihood estimators (MLEs) of the BGGL parameters are remarkably tractable: in the bivariate formulation, the MLEs of $\mu, \delta$, and $\sigma$ correspond to those from the classical least squares regression (without intercept), while the MLEs for the gamma parameters of $X$ are classical and available in closed form. Below are our main contributions: 

\begin{itemize}

\item We derive explicit expressions for the estimators and establish their asymptotic properties as $n \to \infty$. In Theorem \ref{thm:asymp-normal} we prove asymptotic normality for $\alpha > 1$, while Theorems \ref{thm:<1} and \ref{thm:=1} reveal a striking phenomenon: for $\alpha \le 1$, the MLE converges at a faster than classical $\sqrt{n}$ rate. These are the principal results of the paper.

\item It is well known that the Fisher information may be infinite when the underlying density becomes unbounded or irregular (as, for example, in the uniform distribution on $[0,\theta]$). The BGGL family exhibits a related phenomenon - it is a curved exponential family: although parameterized by five parameters, its sufficient statistics span a six-dimensional space. We demonstrate this property.

\item We also explore a financial application, where $Y$ represents stock market returns (e.g., S\&P 500) and $X$ is derived from the Volatility Index (VIX). This setting provides a natural modification of \cite{VIX}, incorporating a stochastic volatility model for the time series $(V, X)$.
\end{itemize}

The paper is organized as follows. Section~\ref{bggl.prop.sec}  introduces the BGGL model and its fundamental properties, including the joint PDF, MGF, characteristic function, connections to L\'evy processes, exponential family representation, Shannon entropy, and Fisher information. Section~\ref{bggl.est.sec} presents the MLE framework. The results on the asymptotic distributions of the MLEs are presented in Section~\ref{sec:normality}. Simulations confirming the distinct convergence behavior for $\alpha \le 1$ versus $\alpha > 1$ are included  in Section~\ref{sec:simulation}. Section~\ref{bggl.financial.sec} illustrates a financial application. Detailed proofs and auxiliary results appear in the Appendix.

\section{Bivariate Gamma Generalized Laplace Distribution: Basic Properties}
\label{bggl.prop.sec}

In this section we formally define the BGGL distribution and account for its basic properties. The random variable $X$ in the definition below follows gamma distribution with shape parameter $\alpha>0$ and scale parameter $\beta>0$, denoted by $\mathrm{GAM}(\alpha, \beta)$ and given by the PDF 
\begin{equation}
\label{eq:Gamma}
f(x) = \frac{x^{\alpha-1}}{\Gamma(\alpha)}\beta^{\alpha}e^{-\beta x},\quad x > 0. 
\end{equation}
%

\begin{defn}
\label{bggl.def}
Let $Y$ admit the stochastic representation (\ref{eq:VG}), where $X\sim \mathrm{GAM}(\alpha,\beta)$ and $Z\sim N(0,1)$ (standard normal), independent of $X$. Then, the random vector $(X,Y)$ is said to have the BGGL distribution with vector-parameter $\boldsymbol \theta = (\alpha, \beta, \delta, \mu, \sigma)$. We denote this distribution by $\mathrm{BGGL}(\alpha, \beta, \delta, \mu, \sigma)$ and write $(X,Y) \sim \mathrm{BGGL}(\alpha, \beta, \delta, \mu, \sigma)$. 
\end{defn}
The notation BGGL stands for {\bf B}ivariate distribution with {\bf G}amma and {\bf G}eneralized {\bf L}aplace marginals. Indeed, by construction, the marginal distribution of $X$ in this bivariate model is gamma distribution. On the other hand, the distribution of $Y$ is GAL ({\bf G}eneralized {\bf A}symmetric {\bf L}aplace), with the PDF given by (see \cite[(4.1.10), (4.1.30)]{Book} or \cite[(1.1), (2.18)]{Survey}):
\begin{equation}
\label{eq:gal.pdf}
f_Y(y)= \frac{\sqrt{2} e^{\frac{\sqrt{2}}{2\tilde{\sigma}}(\frac{1}{\kappa} - \kappa)(y-\delta)}}{\sqrt{\pi}\tilde{\sigma}^{\alpha+\frac{1}{2}}\Gamma(\alpha)}  \left( \frac{\sqrt{2}|y-\delta |}{\kappa + \frac{1}{\kappa}}\right)^{\alpha-\frac{1}{2}}
K_{\alpha-\frac{1}{2}}\left(\frac{\sqrt{2}}{2\tilde{\sigma}}\left( \frac{1}{\kappa} +\kappa \right) |y-\delta | \right), \quad y\in \mathbb R,   
\end{equation}
where $K_{s}(\cdot)$ is the modified Bessel function of the second kind with index $s$. The definition of $K_s$ is taken from \cite[Appendix]{Survey}, or can be found in \cite[Chapter 10]{Handbook}; see the Appendix of this article. In~\eqref{eq:gal.pdf}, we let
\begin{equation}
\label{eq:reparam}
\tilde{\sigma} = \frac{\sigma}{\sqrt{\beta}}, \,\,\, \tilde{\mu} = \frac{\mu}{\beta}, \,\,\, \kappa=\frac{\sqrt{2}\tilde{\sigma}}{\tilde{\mu} + \sqrt{2\tilde{\sigma}^2+\tilde{\mu}^2}} > 0.
\end{equation}
%

Note that while the BGGL distribution is driven by five parameters, the one-dimensional variance-gamma distribution of $Y$ in \eqref{eq:VG} has only four parameters: the parameters $\sigma$ and $\beta$ will not be identifiable. The main properties of GAL distributions are well-known - see, for example, the survey \cite{Survey} and references therein. This random variable has finite moments of all orders, and its moment generating function (MGF) is finite in a neighborhood of zero. The parameter $\mu$ is responsible for skewness: if $\mu = 0$ then the distribution is symmetric. The parameter $\alpha$ is called the {\it shape} and is responsible for the kurtosis. The four parameters allow to pick any mean, variance, skewness and kurtosis, thus giving more flexibility than the classic normal family of distributions. However, as seen above, the density~\eqref{eq:gal.pdf} is complicated as it contains the Bessel special function. The exception is the case $\alpha = 1$, where we obtain AL distribution with an explicit form of the PDF (see, e.g., \cite[(3.0.8)]{Book}. We can then explicitly compute the maximum likelihood estimation (MLE), see the monograph \cite[Section 3.5]{Book}.

\subsection{The joint density} 
By~\eqref{eq:mixture}, the conditional density of $Y\mid X=x$ is Gaussian:
$$
f(y\mid x) = \frac1{\sqrt{2\pi x}\sigma}\exp\left[-\frac{(y - \delta - \mu x)^2}{2\sigma^2x}\right].
$$
Straightforward algebra shows that the 
joint density of $(X,Y)$ takes on an explicit form as follows, for $ x \ge 0$ and $y \in \mathbb R$:
\begin{align}
\label{eq:bggl.pdf}
f(x,y) = f(y \mid x)f(x) = \frac{1}{\sqrt{2\pi}}\frac{1}{\sigma}\frac{\beta^\alpha}{\Gamma(\alpha)} x^{\alpha-3/2}\exp\left[-\left(\beta x+\frac{(y-\delta-\mu x)^2}{2\sigma^2 x} \right)\right].
\end{align}
This can also be written in the form:
\begin{equation}
\label{eq:bgal.pdf.alt}
f(x,y) = \frac{1}{\sqrt{2\pi}}\frac{1}{\sigma}\frac{\beta^\alpha}{\Gamma(\alpha)} x^{\alpha-3/2}\exp\left[-\frac{1}{2}\left\{\frac{\mu^2+2\beta \sigma^2}{\sigma^2} x + \frac{(y-\delta)^2}{\sigma^2}\frac{1}{x} \right\} +\frac{\mu(y-\delta)}{\sigma^2}\right],
\end{equation}
which aids in finding the conditional distribution of $X$ given $Y=y$. Indeed, a close examination of the PDF in (\ref{eq:bgal.pdf.alt}) shows that this conditional distribution is a {\it generalized inverse Gaussian} (GIG) distribution with parameters 
\begin{equation}
\label{gig.par}
a = 2\beta + \frac{\mu^2}{\sigma^2} \ge 0,\quad b = \frac{(y-\delta)^2}{\sigma^2} \ge 0,\quad p = \alpha - 1/2 \in \mathbb R, 
\end{equation}
denoted by $\mathrm{GIG}(a,b,p)$ and given by the PDF 
\begin{equation}
\label{gig.pdf}
f(x) = \frac{(a/b)^{p/2}}{2K_p(\sqrt{ab})} x^{p-1} e^{-\frac{1}{2}\left(ax +\frac{b}{x}\right)}, \quad x \ge 0.
\end{equation}
This fact is crucial in setting-up an estimation algorithm for the parameters of GAL distribution via EM optimization scheme. 

\subsection{The integral transforms} 
\noindent The bivariate moment generating function (MGF) of the BGGL distribution is given by 
\begin{align}
M(s,t) = \mathbb E\left(e^{sX+tY} \right) = e^{\delta t} M_X(s+ \mu t+ 0.5\sigma^2t^2),
\end{align}
where $M_X(u) = \mathbb E[e^{uX}]$ is the MGF of $X$ 
(see Proposition 1 in \cite{Tony}). It is well known, see for example \cite[Appendix B, page 526, Problem 4]{BickelDoksum}, that $M_X(u) = (\beta/(\beta - u))^{\alpha}$, $u<\beta$. This leads to the following result. 

\begin{proposition}
\label{bggl.mgf.prop}
The joint moment generating function of $(X,Y)$ is given by:
\begin{align}
\label{eq:MGF}
\begin{split}
M(s,t) &= e^{\delta t} \beta^{\alpha}\left[\beta - (s+\mu t + 0.5 \sigma^2t^2)\right]^{-\alpha}, \\
 (s,t) & \in D := \{ (s,t)\in \mathbb R^2: s+\mu t + 0.5\sigma^2t^2 < \beta\}. 
\end{split}
\end{align}
\end{proposition}
\noindent The domain $D$ of the MGF in (\ref{eq:MGF}) contains a neighborhood of zero. From here we can see that $\mathbb E[X^k|Y|^l] < \infty$ for all $k, l \ge 0$. We could also compute these moments as partial derivatives of the MGF~\eqref{eq:MGF} at $s = t = 0$:
\begin{equation}
\label{bggl.mom.mgf}
\mathbb E[X^kY^l] = \left.\frac{\partial^{k+l}M}{\partial s^k\partial t^l}\right|_{s=t=0} .
\end{equation}
But it seems to us easier to compute them from scratch using the representation~\eqref{eq:VG}, which we do in the following subsection. The characteristic function (Fourier transform of joint density~\eqref{eq:bgal.pdf.alt} in $\mathbb R^2$) is immediately found from~\eqref{eq:MGF}, with $\mathrm{i} = \sqrt{-1}$:
$$
\varphi(s, t) := \mathbb E\left[e^{\mathrm{i}(sX + tY)}\right] = M(\mathrm{i}s, \mathrm{i}t) = e^{\mathrm{i} \delta t} \beta^{\alpha}\left[\beta - \mathrm{i}s - \mathrm{i}\mu t + 0.5\sigma^2t^2\right]^{-\alpha}.
$$
The BGGL distribution is infinitely divisible, since for any $n \in \mathbb N = \{1, 2, \ldots\}$ the function $\varphi^{1/n}(s, t)$ is also a characteristic function (of a BGGL distribution with parameters $\alpha$ and $\delta$ changed to $\alpha/n$ and $\delta/n$, respectively). In other words, $(X, Y) \sim \mathrm{BGGL}(\alpha, \beta, \delta, \mu, \sigma)$ can be represented as a sum of independent and identically distributed (IID) BGGL random variables as follows:
$$
(X, Y) = (X_1, Y_1) + \cdots + (X_n, Y_n),\quad \mbox{IID}\quad (X_i, Y_i) \sim \mathrm{BGGL}(\alpha/n, \beta, \delta/n, \mu, \sigma).
$$
\subsection{Two-dimensional L\'evy process} 
\noindent Since the BGGL distribution is infinitely divisible, 
one can define a two-dimensional L\'evy process with this bivariate distribution. Let $G = (G(t), t \ge 0)$ be a subordinator (a nondecreasing L\'evy process) built on gamma distribution 
$\textsc{GAM}(\alpha, \beta)$ with shape $\alpha$ and rate $\beta$, so that 
$$
G(t) - G(s) \sim \Gamma(\alpha(t-s), \beta),\quad t > s \ge 0.
$$
Further, let $W = (W(t),\, t \ge 0)$ be a Brownian motion with drift $\mu$ and diffusion $\sigma^2$, independent of $G$, so that 
$$
W(t)  - W(s) \sim \mathcal N(\mu (t-s), \sigma^2(t-s)),\quad t > s \ge 0.
$$
Then the following two-dimensional L\'evy process has BGGL increments:
\begin{equation}
\label{eq:2d-levy}
\{{\bf Y}(t),\, t\ge 0 \} = \{\bigl(G(t), W(G(t)) \bigr),\, t \ge 0\}.
\end{equation}
Indeed, $\bigl(G(t), W(G(t))\bigr) - \bigl(G(s), W(G(s))\bigr)  \sim \mathrm{BGGL}(\alpha(t-s), \beta, 0, \mu, \sigma^2)$. 
This generalizes the construction of the (symmetric or asymmetric) Laplace motion $W(G(t))$ (the second component of~\eqref{eq:2d-levy}) see \cite[Section 4.2]{Book},  especially \cite[Theorem 4.2.1, (4.2.2)]{Book}. For the background, see \cite[Section 4.2]{Survey}. 
\subsection{Moments} 
\noindent The joint moments $\mathbb E[X^kY^l]$, where  $(X,Y)$ follows a BGGL distribution, can be found either using the MGF via (\ref{bggl.mom.mgf}) or via the  stochastic representation (\ref{eq:VG}). Using the latter, we obtain the result below. 

\begin{lemma}
The mean vector and the covariance matrix of $\left(X,Y\right)\sim  \mathrm{BGGL}(\alpha, \beta, \delta, \mu, \sigma)$ are as follows: 
    \begin{equation}
                \begin{bmatrix}
    \frac{\alpha}{\beta} &
    \delta+\mu\frac{\alpha}{\beta}
\end{bmatrix}
\quad \mbox{and}\quad 
   \boldsymbol \Sigma = \begin{bmatrix}
\frac{\alpha}{\beta^2} & \mu\frac{\alpha}{\beta^2} \\
\mu\frac{\alpha}{\beta^2}& \frac{\mu^2\alpha}{\beta^2} + \frac{\sigma^2\alpha}{\beta}
\end{bmatrix}.
    \end{equation}
Moreover, the correlation of $X$ and $Y$ is given by: $\rho = \mu\left[\mu^2+\sigma^2\beta\right]^{-1/2}$.
\end{lemma}
\begin{proof} The mean and variance of marginals for Gamma and Variance-Gamma in~\eqref{eq:VG} are well known. For the Gamma component $X$, from \cite[page 526]{BickelDoksum} we get: 
\begin{equation}
\label{eq:r-moments}
\mathbb E[X^{r}] = 
\begin{cases}
\Gamma(\alpha + r)/(\beta^r\Gamma(\alpha)),\quad r > -\alpha;\\
+\infty,\quad r \le -\alpha.
\end{cases}
\end{equation}
Plugging in $r = 1$ and $r = 2$, or looking at \cite[Appendix B, subsection 2.2]{BickelDoksum}:
\begin{equation}
\label{eq:moments-X}
\mathbb E[X] = \frac{\alpha}{\beta},\quad \mathrm{Var}(X) = \frac{\alpha}{\beta^2},\quad \mathbb E[X^2] = \frac{\alpha + \alpha^2}{\beta^2}.
\end{equation}
For the variance-gamma component $Y$ in~\eqref{eq:VG}, one can see the mean and variance from \cite[subsection 2.7]{VG}. For the first moment, since $\mathbb E[Z] = 0$ and $\mathbb E[Z^2] = 1$, by independence of $X$ and $Z$, and using~\eqref{eq:moments-X}, we obtain 
\begin{equation}
\label{eq:mean-y}
\mathbb E(Y) =  \mathbb E\left( \delta + \mu X + \sigma \sqrt{X} Z\right) = \delta + \mu\cdot \mathbb E X + \sigma \cdot\mathbb E\sqrt{X}\cdot\mathbb E Z = \delta + \mu \frac{\alpha}{\beta}.
\end{equation}
Using~\eqref{eq:moments-X} and~\eqref{eq:mean-y}, we compute the second moment and the variance of $Y$:
\begin{align*}
\begin{split}
\mathbb E \left[Y^2\right]&= \delta^2 + 2\delta \mu \mathbb E \left [ X \right] + \mu^2\mathbb E \left[ X^2 \right] + \sigma^2\mathbb{E}\left[X\right] = \delta^2 + 2 \delta \mu \frac{\alpha}{\beta} + \frac{\alpha^2+\alpha}{\beta^2}\mu^2 +\sigma^2\frac{\alpha}{\beta},\\
    \mathrm{Var}\left[ Y \right] &= \mathbb E \left[ Y^2 \right] - (\mathbb E \left[ Y\right])^2 = \frac{\mu^2\alpha}{\beta^2} + \frac{\sigma^2\alpha}{\beta}.
\end{split}
\end{align*}
Finally, let us compute the covariance. Using~\eqref{eq:VG} and~\eqref{eq:moments-X}, we get:  
\begin{align}
\label{eq:cov.mvcn}
\begin{split}
\mathbb E\left(XY\right) &=  \mathbb E\left[X( \delta + \mu X + \sigma \sqrt{X} Z)\right] = \mathbb E\left[ \delta X +\mu X^2 + \sigma X^{\frac{3}{2}}Z\right]\\
&= \delta \mathbb E\left[X \right] + \mu \mathbb E \left[ X^2 \right ] + \sigma\mathbb E[X^{3/2}]\cdot \mathbb E[Z] = \delta \frac{\alpha}{\beta} + \mu \cdot \frac{\alpha^2 + \alpha}{\beta^2}.
\end{split}
\end{align}
From \eqref{eq:cov.mvcn}, we can calculate the covariance of $X$ and $Y$. Finally, 
$$
\rho = \frac{\mathrm{Cov}(X, Y)}{\sqrt{\mathrm{Var}(X)}\sqrt{\mathrm{Var}(Y)}} = \frac{\mu}{\sqrt{\mu^2+\sigma^2\beta}}.
$$
This completes the proof.
\end{proof}

\begin{Remark}
Note that the correlation, which has the same sign as $\mu$,  vanishes when $\mu=0$ (in which case the distribution of $Y$ is symmetric about zero). However, even though $X$ and $Y$ are uncorrelated in this case, they are not independent. 
\end{Remark}

\subsection{Shannon entropy} 
\noindent In this section we derive {\it Shannon entropy} of 
$$
\left(X,Y\right) \sim  \mathrm{BGGL}(\alpha, \beta, \delta, \mu, \sigma)
$$
(see \cite{Sha48}). For continuous random vectors ${\bf X}$ with the PDF $f(\cdot)$, this measure of uncertainty is defined as the expectation 
\begin{equation}
\label{sha.ent}
H({\bf X}) = - \mathbb E \log f({\bf X}). 
\end{equation}
The result for the BGGL distribution follows from the following proposition, proven in the Appendix, which provides Shannon entropy for any bivariate distribution of $(X,Y)$ where $X$ and $Y$ are independent and $Y$ is conditionally Gaussian as in (\ref{eq:mixture}). 
\begin{proposition}
\label{mix.ent.prop}
If $Y$ is conditionally Gaussian as in (\ref{eq:mixture}), where $X$ is a non-negative random variable independent of $Y$, then 
\begin{equation}
\label{sha.ent.gen}
H((X,Y)) = H(X) +H(Z) +\frac{1}{2}\left( \log \sigma^2 +\mathbb E [\log X] \right),
\end{equation}
where $H(X) = - \mathbb E \log f(X)$ is the Shannon entropy of $X$ and $H(Z) = (1/2)\log(2\pi e)$ is the Shannon entropy of $Z\sim N(0,1)$. 
\end{proposition}
For $\left(X,Y\right)\sim  \mathrm{BGGL}(\alpha, \beta, \delta, \mu, \sigma)$ we have $X\sim GAM(\alpha, \beta)$, in which case we have $\mathbb E [\log X] = \psi(\alpha) - \log \beta$, where $\psi(\cdot) = (\ln \Gamma(\cdot))'$ is the digamma function. Moreover, it is well-known that in this case 
\begin{equation}
\label{sha.ent.gam}
H(X) = \alpha-\log \beta +\log \Gamma(\alpha) +(1-\alpha) \psi(\alpha). 
\end{equation}
In view of these facts, we obtain the result below. 
\begin{corollary}
\label{sha.ent.bggl.coro}
The Shannon entropy of $\left(X,Y\right)\sim  \mathrm{BGGL}(\alpha, \beta, \delta, \mu, \sigma)$ is given by 
\begin{equation}
\label{sha.ent.bggl}
H((X,Y)) = \alpha-\frac{3}{2}\log \beta + \log \Gamma(\alpha) +\left(\frac{3}{2}-\alpha\right) \psi(\alpha) + \frac{1}{2}\log(2\pi e \sigma^2).
\end{equation}
\end{corollary}

\begin{figure}[t]
    \centering
    \subfloat[Density]{\includegraphics[width=8cm]{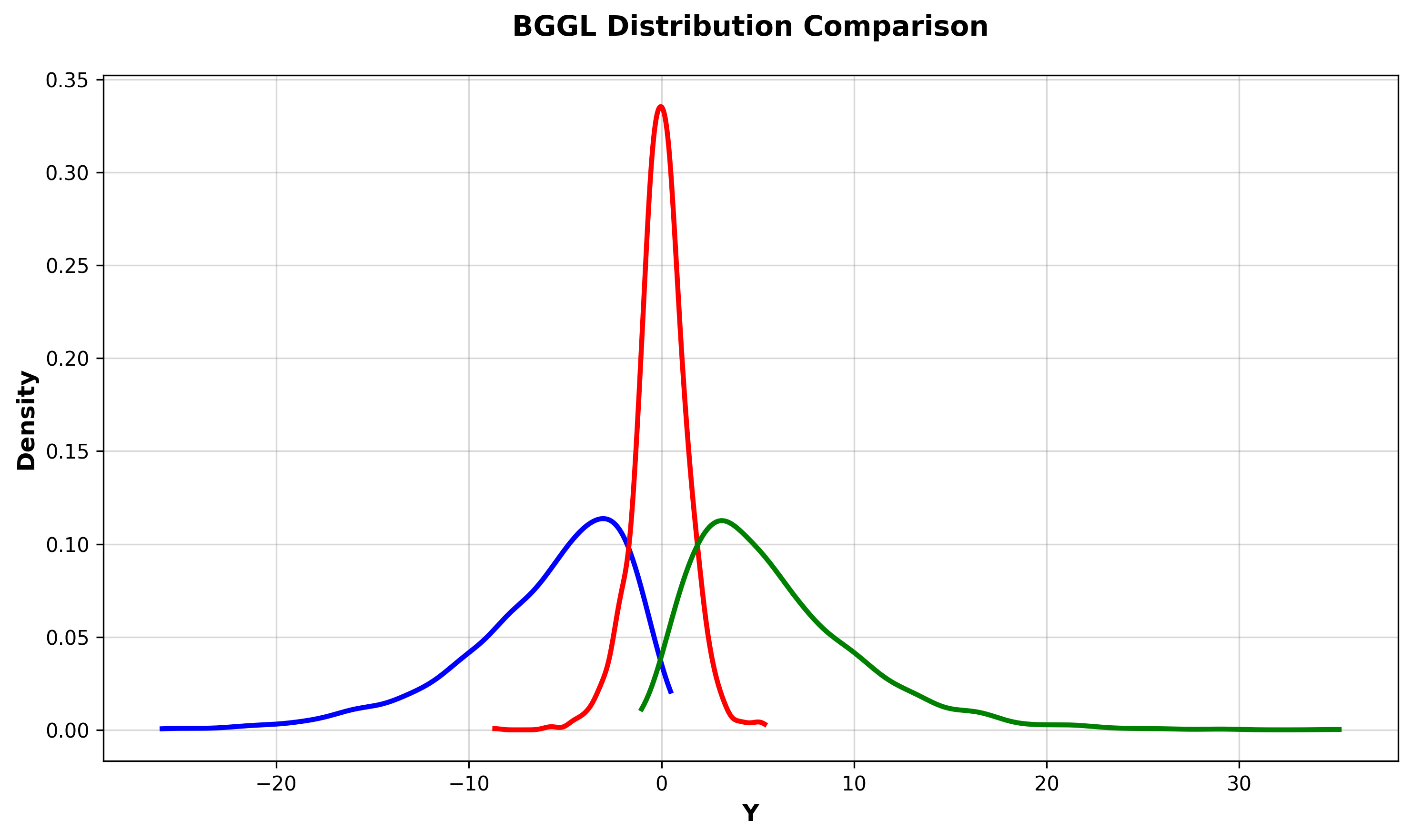}}
      \subfloat[$\mu = 0$]{\includegraphics[width=8cm]{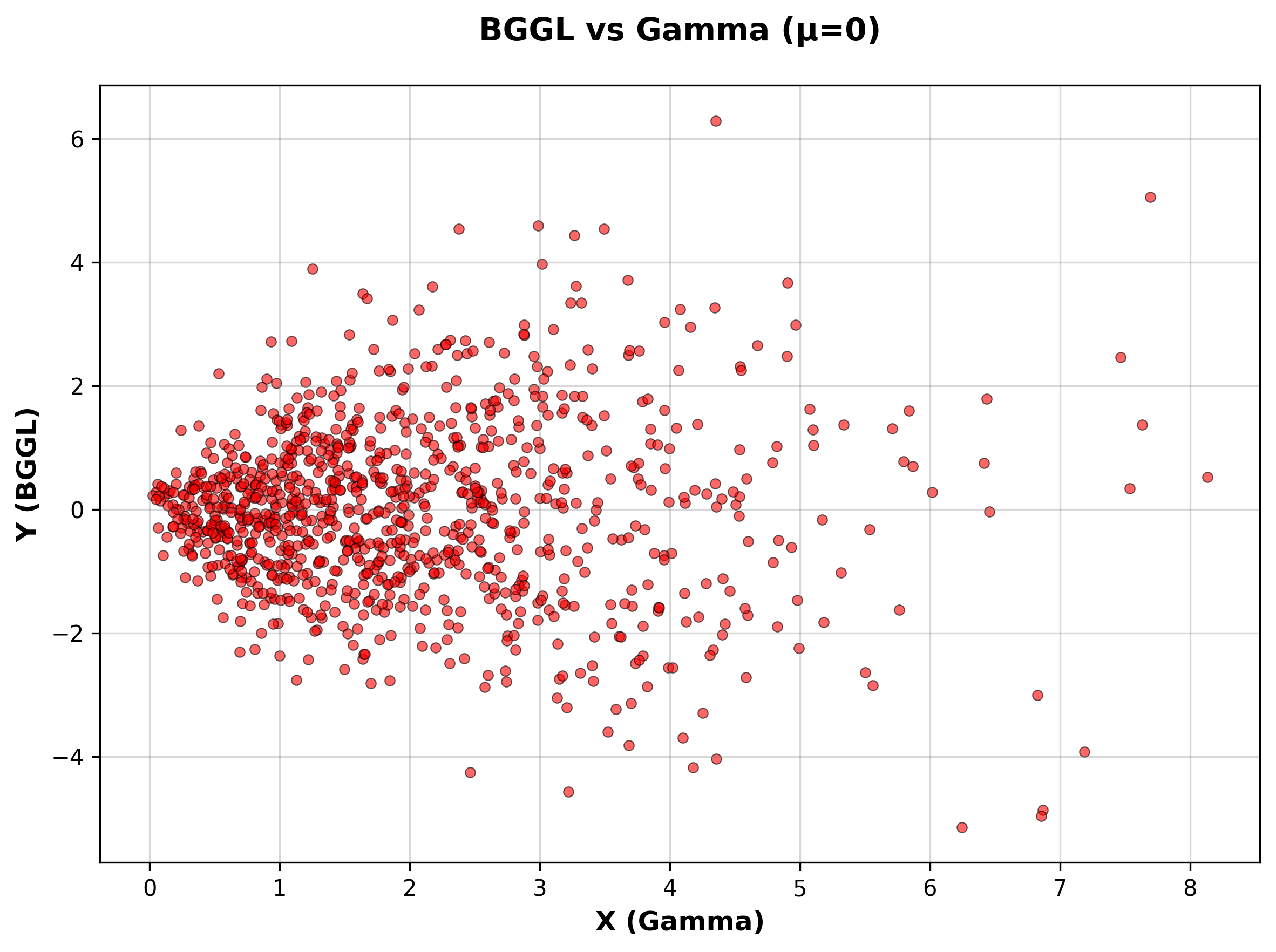}}
    \\
\subfloat[$\mu > 0$]{\includegraphics[width=8cm]{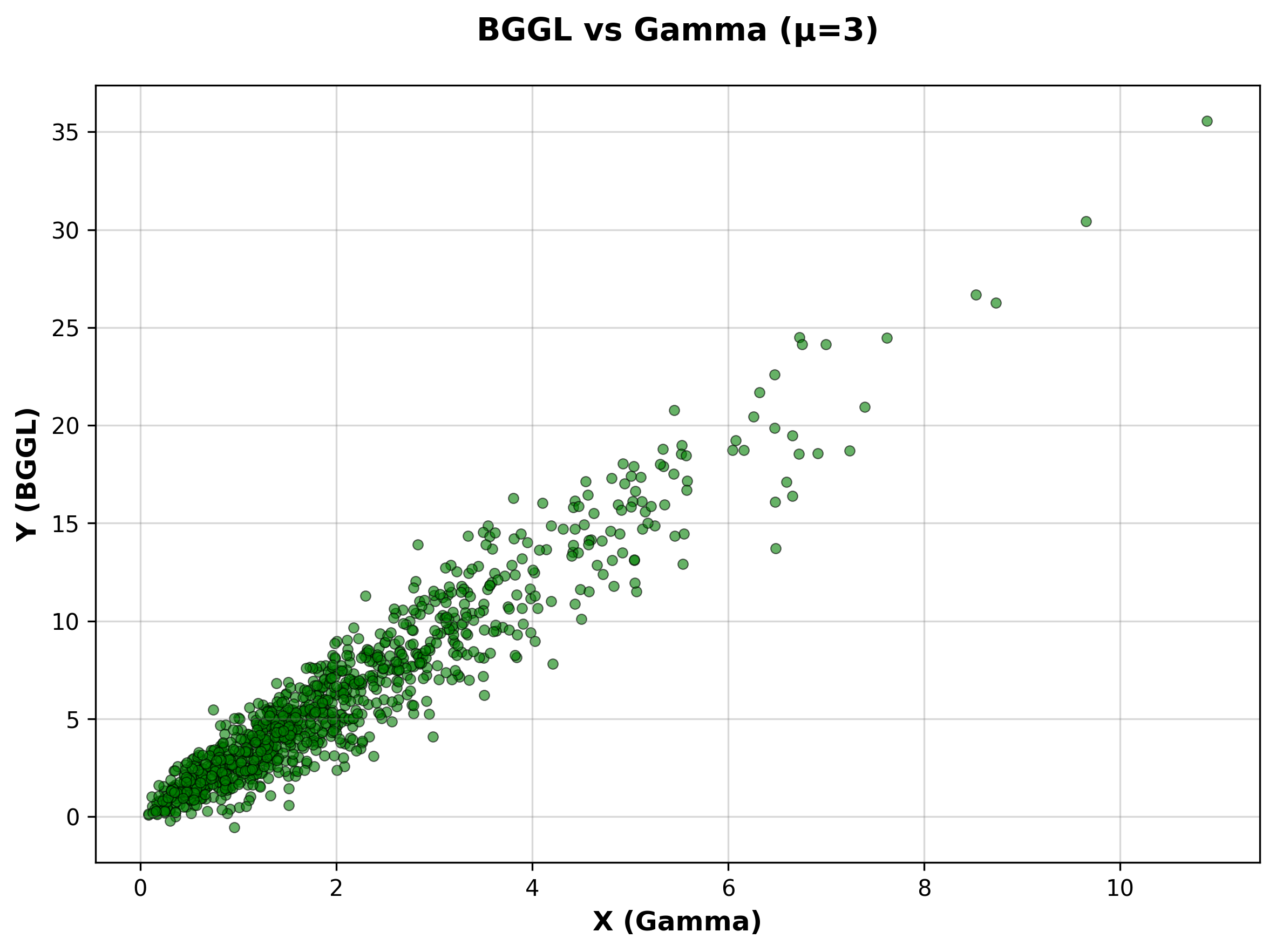}}
\subfloat[$\mu < 0$]{\includegraphics[width=8cm]{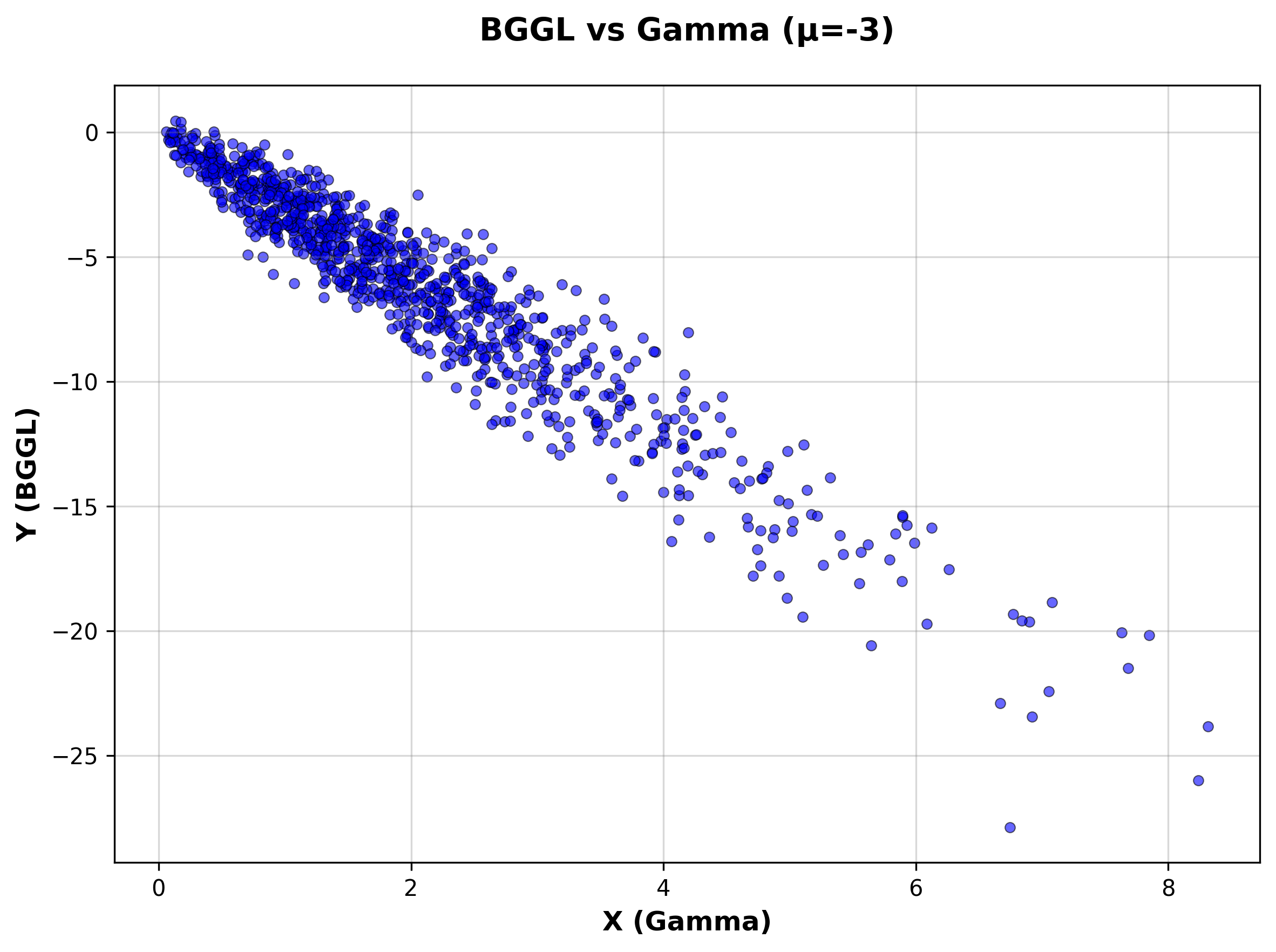}}
     \caption{Scatter plot and density plots of $(X, Y)$. We take $\alpha = 2$, $\beta = 1$, $\delta = 0$, $\mu \in \{\pm 3, 0\}$, $\sigma = 1$.}
    \label{fig:gamma_mu_plots}
\end{figure}

\subsection{Curved exponential family and Fisher information} 
\noindent Straightforward derivations show that the BGGL family given by the PDF~(\ref{eq:bggl.pdf}) is an exponential family. However, it is not a classic family, where the dimension of sufficient statistics is the same as that of the parameter space. Here, the parameter vector has five coordinates:
\begin{equation}
\label{eq:5-params}
\boldsymbol \theta = (\alpha, \beta, \delta, \mu, \upsilon), 
\end{equation}
where, for convenience, we denoted  $\upsilon = \sigma^2$. In this notation, the joint PDF~(\ref{eq:bggl.pdf}) can be written as
\begin{equation}
\label{bggl.exp.fam}
f(x,y) = h(x,y) \exp \left\{\sum_{j=1}^6 \eta_j(\boldsymbol \theta) T_j (x,y) -A(\boldsymbol \theta) \right\},
\end{equation}
where $h(x,y) = (1/\sqrt{2\pi}) x^{-3/2}$ for $x > 0$ and
\begin{equation}
\label{a.theta}
A(\boldsymbol \theta) =-\log \frac{\beta^\alpha}{\Gamma(\alpha)}+\frac{\mu \delta}{\upsilon}+\frac{\log \upsilon}{2}. 
\end{equation}
Further, in (\ref{bggl.exp.fam}) we have 
\begin{align}
\label{eq:6}
\begin{split}
\eta_1(\boldsymbol \theta) = \alpha \quad &\mbox{and}\quad T_1(x, y) = \ln x;\\
\eta_2(\boldsymbol \theta) = -\beta - \frac{\mu^2}{2\upsilon} \quad &\mbox{and}\quad T_2(x, y) = x;\\
\eta_3(\boldsymbol \theta) = -\frac1{2\upsilon} \quad &\mbox{and}\quad T_3(x, y) = \frac{y^2}x;\\
\eta_4(\boldsymbol \theta) = \frac{\delta}{\upsilon} \quad &\mbox{and}\quad T_4(x, y) = \frac yx;\\
\eta_5(\boldsymbol \theta) = -\frac{\delta^2}{2\upsilon} \quad &\mbox{and}\quad T_5(x, y) = \frac1x;\\
\eta_6(\boldsymbol \theta) = \frac{\mu}{\upsilon} \quad &\mbox{and}\quad T_6(x, y) = y.
\end{split}
\end{align}
In general, this is a {\it curved exponential family}, rather than a full one: the dimension of the vector of sufficient statistics, ${\bf T} := (T_1, \ldots, T_6)$,  exceeds the number of parameters in~\eqref{eq:5-params}. See the discussion in \cite[subsection 1.6.3]{BickelDoksum}. Other examples of such distribution families include, for instance, $\{\mathcal N(\mu, \mu^2)\mid \mu \in \mathbb R\}$, as given in \cite[Example 1.6.9]{BickelDoksum}. Moreover, the BGGL family cannot, in general, be written in {\it canonical form} (see \cite[page 54 and Example 3.4.6]{BickelDoksum}). However, in the special case $\delta = 0$, the natural parameters satisfy $\eta_4(\boldsymbol{\theta}) = \eta_5(\boldsymbol{\theta}) = 0$, and the number of parameters matches the number of sufficient statistics (both equal to four). In this case, the model reduces to a full exponential family that admits a canonical representation. See also a useful article \cite{Berk} on asymptotic normality for MLE in exponential models.

\subsection{Fisher information matrix} 
\noindent Under the parametrization in~\eqref{eq:5-params}, the Fisher information matrix $\mathcal{I}$ is a $5\times 5$ matrix representing the covariance matrix of the gradient random vector $\nabla_{\boldsymbol \theta}\log f(X, Y)$ (see, e.g., \cite[Section 3.4.2]{BickelDoksum}):
\begin{equation}
\label{eq:Fisher}
\mathcal I(\boldsymbol \theta) := \mathrm{Cov}\left(\nabla_{\boldsymbol \theta}\log f(X, Y)\right)\quad \mbox{for}\quad (X, Y) \sim \mathrm{BGGL}(\boldsymbol \theta).
\end{equation}
Since $\mathbb E\left(\nabla_{\boldsymbol \theta}\log f(X, Y)\right) = 0$, the $ij$-th element of this matrix is given by 
\begin{equation}
\label{eq:Fisher-matrix}
\mathcal I_{ij}(\boldsymbol \theta) = \mathbb E\left[\frac{\partial \log f(X, Y)}{\partial \theta_i}\frac{\partial\log f(X, Y)}{\partial\theta_j}\right],\quad i, j = 1, \ldots, 5.
\end{equation}
A close examination of the PDF $f(x,y)$ reveals that the terms of $\mathcal{I}$ corresponding to $(\alpha, \beta)$ and 
those corresponding to $(\delta, \mu, \upsilon)$ can be obtained separately, resulting in a block-diagonal matrix. 
In fact, the $2\times 2$ upper block corresponding to $(\alpha, \beta)$ is the same as the Fisher information matrix for the classic gamma family given by the PDF~\eqref{eq:Gamma} (see, e.g., \cite[Chapter 17, subsection 7.2]{Laws}). 
This matrix is known to be:
\begin{equation}
\label{eq:Gamma-Fisher}
\mathcal G(\alpha, \beta) = 
\begin{bmatrix}
\psi'(\alpha) & -\beta^{-1}\\
-\beta^{-1} & \alpha\beta^{-2}\\
\end{bmatrix}.
\end{equation}
Recall that here $\psi(\cdot) = (\ln \Gamma(\cdot))'$ is the digamma function.
Routine derivation shows that the elements $I_{ij}(\boldsymbol \theta)$ for $i = 1, 2$ (corresponding to $\theta_1 = \alpha$ or $\theta_2 = \beta$) and $j = 3, 4, 5$ (corresponding $\theta_3 = \delta, \theta_4 = \mu, \theta_5 = \upsilon$) are zero, and the Fisher information matrix is block-diagonal with two blocks: an upper-left block,  corresponding to $(\alpha, \beta)$, and a lower-right block,  corresponding to $(\delta, \mu, \upsilon)$. The calculation of the $3\times 3$ lower-right block corresponding to $(\delta, \mu, \upsilon)$ is straightforward via (\ref{eq:Fisher-matrix}), where we have:

\begin{align}
\label{eq:3-grad}
\begin{split}
\frac{\partial \ln f(x, y)}{\partial \delta} &= -\frac{\mu}{\upsilon} + \frac{y}{\upsilon x} - \frac{\delta}{\upsilon x},\\
\frac{\partial \ln f(x, y)}{\partial \mu} &= -\frac{\mu}{\upsilon}x +\frac{y}{\upsilon}  -\frac{\delta}{\upsilon},\\
\frac{\partial \ln f(x, y)}{\partial \upsilon} & = \frac{(y-\delta - \mu x)^2}{2\upsilon^2 x} - \frac{1}{2\upsilon}.
\end{split}
\end{align}
Recall that $\upsilon = \sigma^2$. We shall omit routine derivations leading to the result below. 
\begin{theorem}
\label{bggl.fisher.theo}
If $\alpha > 1$, the Fisher information matrix for the $\mathrm{BGGL}(\boldsymbol \theta)$  is
\begin{equation}
\label{eq:5D}
\mathcal I(\boldsymbol \theta) = 
\begin{bmatrix}
\psi'(\alpha) &  -1/\beta  & 0 & 0 & 0\\
-1/\beta & \alpha/\beta^2 & 0 & 0 & 0\\
0 & 0 & \beta(\alpha - 1)^{-1}\sigma^{-2} & \sigma^{-2} & 0\\
0 & 0 & \sigma^{-2} & \alpha\beta^{-1}\sigma^{-2} & 0 \\
0 & 0 & 0 & 0 & 0.5\sigma^{-4} \\
\end{bmatrix} .
\end{equation}
\end{theorem}

The restriction $\alpha>1$ in Theorem \ref{bggl.fisher.theo} is related to the derivation of the diagonal element corresponding to $\delta$.  Indeed, a routine derivation utilizing (\ref{eq:VG}) and (\ref{eq:3-grad}) leads to

\begin{equation}
\label{eq:condition}
\frac{\partial \ln f(X, Y)}{\partial \delta } = \frac{Y - \mu X - \delta}{\upsilon X} = \frac{\sigma\sqrt{X}Z}{\upsilon X} = \frac{Z}{\sigma\sqrt{X}}.
\end{equation}
Since $Z \sim \mathcal N(0,1)$ has finite second moment and is independent of $X$, the right-hand side of~\eqref{eq:condition}  has finite second moment if and only if $\mathbb E[(1/\sqrt{X})^2] = \mathbb E[X^{-1}] < \infty$. Plugging in $r = -1$ into \cite[page 526, exercise 4(a)]{BickelDoksum} and using  $\Gamma(\alpha) = (\alpha - 1)\Gamma(\alpha-1)$, we conclude that 
\begin{equation}
\label{exp.inv.gam}
\mathbb E[X^{-1}] = 
\begin{cases}
\beta/(\alpha - 1),\quad \alpha > 1;\\
+\infty,\quad \alpha \in (0, 1].
\end{cases}
\end{equation}
Therefore, the gradient vector $\nabla_{\boldsymbol \theta}\log f(X, Y)$ has finite second moment if and only if $\alpha > 1$. 


The derivation of the $3\times 3$ bottom-right block of the Fisher information matrix~\eqref{eq:Fisher} corresponding to $(\delta, \mu, \upsilon)$ is quite general, and can be related to a linear model with stochastic covariates (see, e.g., \cite[Example 6.2.1]{BickelDoksum}). Routine algebra shows that this block is given by $\upsilon^{-1}\mathrm{diag}(\mathcal H, 0.5\upsilon^{-1})$, 
where 
\begin{equation}
\label{eq:Fisher-linear}
\mathcal H := 
\begin{bmatrix}
\mathbb E\left(X^{-1/2}\cdot X^{-1/2}\right) & \mathbb E\left(X^{1/2}\cdot X^{-1/2}\right) \\
\mathbb E\left(X^{1/2}\cdot X^{-1/2}\right) & \mathbb E\left(X^{1/2}\cdot X^{1/2}\right)
\end{bmatrix}
= \begin{bmatrix}
\mathbb E[X^{-1}] & 1\\
1 & \mathbb E[X]
\end{bmatrix},
\end{equation}
provided that the required expectations in (\ref{eq:Fisher-linear}) are finite. Since in the gamma case we have $\mathbb E[X] = \alpha/\beta$ and, by (\ref{exp.inv.gam}),  $\mathbb E[1/X]  = \beta/(\alpha-1)$ (for $\alpha>1$), we recover the lower-right block of $\mathcal{I}$ in Theorem \ref{bggl.fisher.theo}. See also \cite{Infinite} for another case where Fisher information is infinite and the convergence rate is different than $n^{1/2}$. 

\section{Parameter Estimation}
\label{bggl.est.sec}

\noindent 
Given a bivariate random sample $(X_i, Y_i)$, $i = 1, \ldots, n$, from $BGGL(\alpha, \beta, \delta, \mu, \upsilon)$ given by the PDF~\eqref{eq:bgal.pdf.alt}, where, as before, $\upsilon = \sigma^2$, the {\it log-likelihood} (LL) function $\ell(\alpha, \beta, \delta, \mu, \upsilon)$ 
%
%
can be easily found to be 
\begin{equation}
\label{bggl.ll.fun}
\ell = n\left[C + g(\delta, \mu, \upsilon) + h(\alpha, \beta)\right],
\end{equation}
where we define
\begin{align}
\label{eq:bgal.ll.fun}
\begin{split}
C & = -\frac12\log(2\pi) - \frac{3}{2}\cdot \overline{\log X},\\
g(\delta, \mu, \upsilon) &= -\frac{1}{2}\cdot\log\upsilon - \frac1n\sum\limits_{i=1}^n\frac{(Y_i-\delta)^2}{2\upsilon X_i} - \frac{\mu^2}{2\upsilon}\overline{X} + \frac{\mu}{\upsilon}(\overline{Y}-\delta),\\
h(\alpha, \beta) &= \alpha\log\beta - \log\Gamma(\alpha) + \alpha \cdot\overline{\log X} - \beta\cdot\overline{X},
\end{split}
\end{align}
and the notation $\overline{U}$ denotes the empirical average of the $\{U_i\}$. Clearly, the search for the maximum likelihood estimators (MLEs) can be carried out for $\delta$, $\mu$, and $\upsilon$ separately from $\alpha$ and $\beta$ via maximization of the functions $g$ and $h$ in (\ref{eq:bgal.ll.fun}), respectively. 


\subsection{The MLEs of $\alpha$ and $\beta$} 

\noindent The MLEs of the parameters $\alpha, \beta$ of the gamma mixing distribution, which arise by maximizing the function $h$ in (\ref{eq:bgal.ll.fun}), are well known; see for example \cite[Chapter 17, subsection 7.2]{Laws}. It is easy to see that, for any fixed $\alpha>0$, the function $h$ in (\ref{eq:bgal.ll.fun}) is maximized by a unique value 

\begin{equation}
\label{eq:beta}
\beta(\alpha) = \frac{\alpha}{\overline{X}}.
\end{equation}
A substitution of this $\beta(\alpha)$ into the function $h$ results in the following function of $\alpha$:
\begin{equation}
\label{eq:upsilon}
h(\alpha, \beta(\alpha)) = \alpha\log\alpha -\alpha\log\overline{X} - \log\Gamma(\alpha) + \alpha \, \overline{\log X} - \alpha.
\end{equation}
The function~\eqref{eq:upsilon} needs to be maximized with respect to $\alpha$. 
However, the derivative of this function, $(d/d\alpha)h(\alpha, \beta(\alpha)) = \log(\alpha) - \psi(\alpha) - (\log \overline{X} - \overline{\log X})$, is a continuous, strictly decreasing function on $(0,\infty)$, with the limits of $\infty$ and  $- (\log \overline{X} - \overline{\log X})$ at zero and infinity, respectively (see, e.g., Lemma 3 in \cite{AKP21}). Moreover, by the concavity of the logarithmic function, the limit at infinity is negative (unless all sample values are the same, which occurs with probability 0). Consequently, with probability 1, there exist a unique MLE of the parameter $\alpha$, which satisfies the likelihood equation

\begin{equation}
\label{eq:alpha}
w(\alpha) = \log(\alpha) - \psi(\alpha) = \log\overline{X} - \overline{\log X}
\end{equation}
and requires (straightforward) numerical methods to solve. In conclusion, the MLEs of $\alpha$ and $\beta$ exist and are unique (with probability 1), and admit the stochastic representation  
\begin{equation}
\label{gam.mles}
\begin{bmatrix}
            \hat{\alpha}_n \\ \hat{\beta}_n
        \end{bmatrix} = 
        \begin{bmatrix}
1 & 0\\
0 & 1/\overline{X}\\
\end{bmatrix}
        \begin{bmatrix}
   w^{-1}(\log\overline{X} - \overline{\log X})\\
   w^{-1}(\log\overline{X} - \overline{\log X})
\end{bmatrix},
\end{equation}
where $w^{-1}(\cdot)$ is the inverse function of $w(\cdot)$ given by the left-hand side in (\ref{eq:alpha}). 

\subsection{The MLEs of $\delta$, $\mu$, and $\sigma^2$} 
\noindent The MLEs of the parameters  $\delta$, $\mu$, and $\upsilon = \sigma^2$ are easily obtained by maximization of the function $g(\delta, \mu, \upsilon)$ in (\ref{eq:bgal.ll.fun}), leading to explicit values presented in the result below, proven in the appendix. 
\begin{theorem}
\label{bggl.mle.theo}
Let $(X_1, Y_1), \ldots, (X_n,Y_n)$ be a random sample from $BGGL(\alpha, \beta, \delta, \mu, \upsilon)$ distribution given by the PDF~\eqref{eq:bgal.pdf.alt} where $\upsilon = \sigma^2$ and $n\geq 2$. Then, the MLEs of $\delta$,  $\mu$, and $\upsilon$ exist and are unique with probability one, and are given by 
\begin{equation}
\label{eq:delta-mu}
 \hat{\delta}_n = \frac{\overline{Y}\cdot\overline{X}^{-1}-\overline{(X^{-1}Y)}}{(\overline{X})^{-1} -\overline{X^{-1}}},\quad 
\hat{\mu}_n = \frac{\overline{Y}-\hat{\delta}_n}{\overline{X}}, 
\end{equation}
\begin{equation}
\label{eq:s2}
\hat{\upsilon}_n = \frac1n\sum\limits_{i=1}^n\left[X_i^{-1/2}Y_i - X_i^{-1/2}\hat{\delta}_n - X_i^{1/2}\hat{\mu}_n\right]^2.
\end{equation}
\end{theorem}

If all the values of $X_i$ are equal to the same $X>0$, which occurs with probability 0, then the three MLEs are not unique. This can be seen from the arguments given in the last section of the Appendix, where other special cases are discussed as well.

Further properties of the estimators are established in the result below. These properties are well known in regression settings with error variance proportional to the explanatory variables, as discussed briefly at the end of this section.

\begin{theorem}
\label{thm:unbiased}
The MLEs $\hat{\delta}_n$ and $\hat{\mu}_n$ in~\eqref{eq:delta-mu} are unbiased. Moreover, the distribution of $(\hat{\delta}_n, \hat{\mu}_n)$, conditional on $X_1, \ldots, X_n$, is bivariate Gaussian:
$$
\label{eq:bivariate-normal}
\begin{bmatrix}\hat{\delta}_n \\ \hat{\mu}_n\end{bmatrix} \sim \mathcal N_2\left(\begin{bmatrix}\delta \\ \mu\end{bmatrix}, \sigma^2\begin{bmatrix} \sum X_k^{-1} & n\\ n & \sum X_k\end{bmatrix}^{-1}\right) .
$$
In addition, while the estimator $\hat{v}_n$ is biased for $\sigma^2 = v$, it is asymptotically unbiased. Finally, 
\begin{equation}
\label{s.square}
s^2 = \frac{n}{n-2}\hat{v}_n = \frac1{n-2}\sum\limits_{i=1}^n\left[X_i^{-1/2}Y_i - X_i^{-1/2}\hat{\delta}_n - X_i^{1/2}\hat{\mu}_n\right]^2
\end{equation}
is an unbiased estimator of $\sigma^2$ and $(n-2)s^2 \sim \chi^2_{n-2}$.
\end{theorem}

Let us note that the maximum likelihood estimators of $\mu$ and $\delta$ coincide with least squares estimators. To see this, consider the representation
\begin{equation}
\label{eq:regression}
Y_i = \delta + \mu X_i + \sigma \sqrt{X_i}\, Z_i, 
\quad Z_i \sim \mathcal{N}(0,1), 
\quad i = 1, \ldots, n.
\end{equation}
This model can be viewed as a \emph{generalized} (or \emph{weighted}) least squares problem, where $\mu$ and $\delta$ are chosen to minimize
\begin{equation}
\label{eq:GLS}
\sum_{i=1}^n \frac{\bigl(Y_i - \delta - \mu X_i\bigr)^2}{X_i}.
\end{equation}
See, for example, \cite[p.~112]{BickelDoksum} for background on weighted least squares. Equivalently, we may transform \eqref{eq:regression} to obtain a model with homoscedastic errors by dividing both sides by $\sqrt{X_i}$:
\begin{equation}
\label{eq:reg}
X_i^{-1/2} Y_i 
= X_i^{-1/2} \delta + X_i^{1/2} \mu + \sigma Z_i, 
\quad i = 1, \ldots, n.
\end{equation}
This yields a linear regression model with constant error variance. Although the transformed model does not include a standard intercept term, this does not affect its statistical properties. Applying ordinary least squares (OLS) to \eqref{eq:reg} leads to minimization of
\begin{equation}
\label{eq:OLS}
\mathcal{S}(\mu, \delta) 
= \sum_{i=1}^n 
\left[
X_i^{-1/2} Y_i 
- X_i^{-1/2} \delta 
- X_i^{1/2} \mu
\right]^2.
\end{equation}
For further discussion of OLS, see the classic textbook \cite{BickelDoksum}, Section~2.2, Theorem~6.1.4, Corollary~6.1.1, and Proposition~6.1.1.

\section{Asymptotic Distributions of the Estimators} 

\label{sec:normality}
Consider first the case $\alpha > 1$, where the MLE $\hat{\boldsymbol \theta}$ is asymptotically normal and efficient, see \cite[Section 6.2.2]{BickelDoksum}. Recall the notation $\upsilon = \sigma^2$.
\begin{theorem}
\label{thm:asymp-normal}
If $\alpha > 1$, we have weak convergence to a 5-dimensional Gaussian distribution:
$$
\sqrt{n}(\hat{\boldsymbol \theta} - \boldsymbol  \theta)^\top \stackrel{d}{\to} \mathcal N_5(\mathbf{0}, \mathcal I^{-1}(\boldsymbol  \theta)),\quad n \to \infty,
$$
where $\mathcal I(\boldsymbol  \theta)$ is the Fisher information matrix provided in Theorem \ref{bggl.fisher.theo}.
\end{theorem}
Since $\mathcal I(\boldsymbol \theta)$ consists of three blocks, the same can be said about its inverse $\mathcal I^{-1}(\boldsymbol  \theta)$: It has blocks of size $2\times2, 2\times 2, 1\times 1$ corresponding to $(\alpha, \beta)$, $(\delta, \mu)$ and $\upsilon$. Therefore, $(\hat{\alpha}_n, \hat{\beta}_n)$,  $(\hat{\delta}_n, \hat{\mu}_n)$ and $\hat{\upsilon}_n$ are asymptotically independent. This observation leads to the following three corollaries.
\begin{corollary} 
\label{cor:alpha-beta}
For $\alpha>1$, the asymptotic distribution of $\hat{\alpha}_n$ and $\hat{\beta}_n$ is bivariate Gaussian:
$$
\sqrt{n}\begin{bmatrix}\hat{\alpha}_n  - \alpha & \hat{\beta}_n - \beta\end{bmatrix}^\top \stackrel{d}{\to} \mathcal N_2\left(\mathbf{0}, 
\boldsymbol \Sigma_{\alpha, \beta}  \right),\quad n \to \infty,
$$
where the limiting covariance matrix (recall  the digamma function $\psi(\cdot) = (\ln \Gamma(\cdot))'$) 
\begin{equation}
\label{eq:Sigma-alpha-beta}
\boldsymbol \Sigma_{\alpha, \beta} = 
[\alpha\psi'(\alpha) -1]^{-1}\begin{bmatrix}
\alpha & \beta\\
\beta & \beta^{2}\psi'(\alpha)\\
\end{bmatrix}
\end{equation}
is the inverse of the $2\times 2$-matrix $\mathcal G(\alpha, \beta)$ given by ~\eqref{eq:Gamma-Fisher}.
\end{corollary}


\begin{corollary} 
\label{cor:delta-mu}
For $\alpha>1$, the asymptotic distribution of $\hat{\delta}_n$ and $\hat{\mu}_n$ is bivariate Gaussian:
$$
\sqrt{n}\begin{bmatrix}\hat{\delta}_n  - \delta & \hat{\mu}_n - \mu \end{bmatrix}^\top \stackrel{d}{\to} \mathcal N_2\left(\mathbf{0}, \boldsymbol \Sigma_{\delta, \mu} \right),\quad n \to \infty,
$$
with the limiting covariance matrix
\begin{equation}
\label{eq:Sigma-delta-mu}
\boldsymbol \Sigma_{\delta, \mu} = 
\frac{\sigma^2}{\beta}
\begin{bmatrix}
\alpha(\alpha-1) & -\beta(\alpha-1)\\
-\beta(\alpha-1) &\beta^2\\
\end{bmatrix}
\end{equation}
\end{corollary}

\begin{corollary} 
\label{cor:upsilon}
For $\alpha>1$, the asymptotic distribution of $\hat{\upsilon}_n$ is:
\begin{equation}
\label{eq:variance}
\sqrt{n}(\hat{\upsilon}_n - \upsilon) \stackrel{d}{\to} \mathcal N(0, 2\sigma^4),\quad n \to \infty.
\end{equation}
\end{corollary}
\begin{Remark} 
As shown below, the asymptotic distribution of $\hat{\alpha}_n$ and $\hat{\beta}_n$ as well as that of $\hat{\upsilon}_n$, given in Corollaries \ref{cor:alpha-beta} and \ref{cor:upsilon}, respectively,  carry over to the case $\alpha \in (0,1]$. However, this is no longer so for $\hat{\delta}_n$ and $\hat{\mu}_n$. 
\end{Remark}

Note that the asymptotic variance of $\hat{\delta}_n$ given in Corollary \ref{cor:delta-mu} becomes zero for $\alpha=1$. This is related to the fact that the corresponding entry of the Fisher information matrix~\eqref{eq:Fisher} is infinite in this case. However, under a different normalization, the asymptotic distribution of $\hat{\delta}_n$ and $\hat{\mu}_n$ is still bivariate normal, with the two estimators becoming asymptotically independent. On the other hand,  for $\alpha<1$, the Gaussianity of the asymptotic distribution breaks down. These facts are presented in the two results stated below, starting with the case $\alpha \in (0,1)$. 

\begin{theorem} 
For $\alpha \in (0,1)$ we have:
%
\begin{equation}
\label{stef1}
{\bf D}_n(\hat{\boldsymbol \theta}_n - \boldsymbol \theta)^\top \stackrel{d}{\to} 
\mathbf{W}^\top = (W_{\alpha}, W_{\beta}, W_{\delta}, W_{\mu}, W_{\upsilon})^\top,\quad n \to \infty,
\end{equation}
where 
\begin{equation}
\label{eq:diagonal-matrix}
{\bf D}_n = \mathrm{diag}(\sqrt{n}, \sqrt{n}, n^{1/(2\alpha)}, \sqrt{n}, \sqrt{n})
\end{equation}
is a $5 \times 5$ diagonal matrix and the limiting vector ${\bf W}$ has four independent components:
\begin{align}
\label{eq:all-W}
\begin{split}
(W_{\alpha}, W_{\beta})^\top &\sim \mathcal N_2(\mathbf{0}, \boldsymbol \Sigma_{\alpha, \beta}),\\
W_{\delta} &\stackrel{d}{=} \sigma\beta^{-1/2}[\Gamma(\alpha+1)]^{1/2\alpha}\xi_{\alpha}^{-1/2}Z,\\
W_{\mu} &\sim \mathcal N(0, \sigma^2\beta/\alpha),\\ 
W_{\upsilon} &\sim \mathcal N(0, 2\sigma^4).
\end{split}
\end{align}
Here, the matrix $\boldsymbol\Sigma_{\alpha, \beta}$ is taken from~\eqref{eq:Sigma-alpha-beta}, the two random variables $Z \sim \mathcal N(0, 1)$ and $\xi_\alpha$ are independent, and $\xi_\alpha$ is a stable subordinator defined via its Laplace transform: 

\begin{equation}
\label{stable.sub.lt}
\mathbb E \left[e^{-u \xi_{\alpha}}\right] = \exp\left[-\Gamma(1 - \alpha)u^{\alpha}\right],\quad u \ge 0.
\end{equation}
\label{thm:<1}
\end{theorem}
For the case $\alpha = 1$, we have a different scaling ($\sqrt{n\ln n}$ instead of $n^{1/(2\alpha)}$), and a different limiting distribution connected with $\hat{\delta}_n$, as detailed in the result below. 

\begin{theorem} 
\label{thm:=1}
For $\alpha = 1$ we have: 
\begin{equation}
\label{stefalpha=1}
{\bf D}_n(\hat{\boldsymbol \theta}_n - \boldsymbol \theta)^\top  \stackrel{d}{\to} 
{\bf W}^\top = (W_{\alpha}, W_{\beta}, W_{\delta}, W_{\mu}, W_{\upsilon})^\top,  \quad n \to \infty,
\end{equation}
where, instead of~\eqref{eq:diagonal-matrix}, this time the diagonal matrix ${\bf D}_n$ is given by 
$$
{\bf D}_n = \sqrt{n}\cdot \mathrm{diag}(1, 1, \sqrt{\ln n}, 1, 1)
$$
and, instead of~\eqref{eq:all-W}, the limiting vector $\mathbf{W}$ has four independent components:
\begin{align}
\label{eq:=1}
\begin{split}
(W_{\alpha}, W_{\beta})^\top &\sim \mathcal N_2(\mathbf{0}, \boldsymbol\Sigma_{1, \beta}),\\
W_{\delta} &\sim \mathcal N(0, \sigma^2/\beta),\\
W_{\mu} &\sim \mathcal N(0, \sigma^2\beta),\\
W_{\upsilon} &\sim \mathcal N(0, 2\sigma^4),
\end{split}
\end{align}
where $\boldsymbol\Sigma_{1, \beta}$ is the same matrix as in~\eqref{eq:Sigma-alpha-beta} with $\alpha = 1$. 
\end{theorem}

\begin{Remark} 
In all of the above results, the MLE $\hat{\upsilon}_n$ may be replaced by the classical estimator $s^2$ defined in~\eqref{s.square}. 
\end{Remark}

\begin{Remark} We emphasize that, in all three cases, the estimators $(\hat{\alpha}_n, \hat{\beta}_n)$, $(\hat{\delta}_n, \hat{\mu}_n)$, and $\hat{\upsilon}_n$ are asymptotically independent. Moreover, in the nonstandard case $\alpha \leq 1$, the estimators $\hat{\delta}_n$ and $\hat{\mu}_n$ are also asymptotically independent. In contrast, in the regular case $\alpha > 1$, these two estimators are not asymptotically independent, as the corresponding off-diagonal element of the limiting covariance matrix is nonzero. Finally, for $\alpha \geq 1$, the limiting distribution is multivariate Gaussian, whereas for $\alpha < 1$ it is not, as it involves a stable subordinator.
\end{Remark}
%
\begin{Remark}
The limiting variance $\sigma^2_{\infty,\mu}$ of the (asymptotically normal) estimator $\hat{\mu}_n$ depends nonlinearly on $\alpha$. For $\alpha > 1$, it follows from Corollary~\ref{cor:delta-mu} as the second diagonal element of the limiting covariance matrix~\eqref{eq:Sigma-delta-mu}. For $\alpha = 1$, it is obtained from Theorem~\ref{thm:=1} via the third relation in~\eqref{eq:=1}. Finally, for $\alpha < 1$, it follows from Theorem~\ref{thm:<1}, again from the third relation in~\eqref{eq:all-W}. Collecting these cases, we obtain
\[
\sigma^2_{\infty,\mu} = \frac{\sigma^2 \beta}{\min(\alpha,1)}.
\]
An intuitive explanation is as follows. The joint finite-sample distribution of $(\hat{\delta}_n,\hat{\mu}_n)$ is given in Theorem~\ref{thm:unbiased}, from which the asymptotic behavior of $\hat{\mu}_n$ can be deduced as $n \to \infty$. For $\alpha > 1$, both diagonal elements of the inverse covariance matrix are of order $n$, and thus the limiting distribution of $(\hat{\delta}_n,\hat{\mu}_n)$ is not of product form. In contrast, for $\alpha \le 1$, the term $\sum_{i=1}^n X_i^{-1}$ grows faster than $n$, which leads to asymptotic independence of $\hat{\delta}_n$ and $\hat{\mu}_n$. Consequently, the limiting variance $\sigma^2_{\infty,\mu}$ is determined independently of the asymptotic behavior of $\hat{\delta}_n$.
\end{Remark}

There are many well-documented cases in the literature where maximum likelihood estimators do not exhibit the usual asymptotic normality. In particular, such behavior may arise when the support of the distribution depends on the unknown parameters. A classical example is the uniform distribution on $[0,\theta]$, where $\theta>0$ is unknown. In this case, the MLE of $\theta$ is the sample maximum $X_{(n)}$, and it is well known that $n\bigl(\theta - X_{(n)}\bigr)$  converges in distribution to an exponential random variable. Thus, the convergence rate is $n$, rather than the usual $\sqrt{n}$. Another example arises in location families of the form $g(x;\mu) = f(|x-\mu|)/2$, where $x, \mu \in \mathbb{R}$ and $f(\cdot)$ is a PDF on $\mathbb R_+$ with $f(0+)=\infty$. In this case, the MLE of $\mu$ may fail to exist; see \cite{Hossain,Wallis}. Indeed, the density $g(\cdot; \mu)$ has a singularity at $x=\mu$, which leads to irregular behavior. Related phenomena are discussed in \cite{Infinite}. In the BGGL model, we observe a similar effect for $\alpha<1$: the convergence rate deviates from $\sqrt{n}$, and the limiting distribution is no longer Gaussian.

\section{Simulation} 
\label{sec:simulation}
To illustrate the behavior of the maximum likelihood estimators (MLEs) for the parameters $\delta$, $\sigma$, and $\mu$, we conducted a simulation study. Samples of size $n=50$ and $n=500$ were generated by first simulating Gamma and standard normal random variables, and then constructing generalized asymmetric Laplace (GAL) observations using fixed parameter values of $\delta$, $\sigma$, and $\mu$. For each generated dataset, the MLEs of the parameters were computed. This procedure was repeated $5000$ times for each sample size.

To assess the performance of the estimators, we report, for each parameter, the empirical mean, variance, and root mean squared error (RMSE) based on the simulated samples.

We consider four representative values of the shape parameter $\alpha$: $\alpha=1$, corresponding to the exponential distribution; $\alpha=0.25$, where the Gamma density is unbounded at zero; $\alpha=2$, representing an intermediate case; and $\alpha=5$, for which the distribution is closer to normality. In all cases, the scale parameter is fixed at $\beta=1$.

The simulation results are summarized in Table~\ref{table:sim}. Overall, the estimators of $\mu$, $\delta$, and $\sigma^2=\upsilon$ perform very well for small values of $\alpha$, which is consistent with the classical setting of ordinary least squares regression. However, for $\alpha=2$ and especially for $\alpha=5$, the RMSE associated with $\delta$ becomes noticeably larger.

These findings are consistent with the theoretical results presented in Section~5. In particular, Theorem~\ref{thm:asymp-normal} establishes a $\sqrt{n}$ convergence rate for $\alpha>1$, while Theorems~\ref{thm:<1} and~\ref{thm:=1} show that for $\alpha \in (0,1)$ and $\alpha=1$, respectively, the convergence rate is faster than $\sqrt{n}$.

\begin{table}[htbp]
\centering
\begin{tabular}{|c|c|c|c|c|}
\hline
\multicolumn{5}{|c|}{Sample size $n = 50$, shape $\alpha = 1$} \\
\hline
 & Actual Value & Mean Estimated Value & Variance  & RMSE \\
\hline
$\delta$ & 1& .9993& .0214& .1105\\
\hline
$\sigma$ & 2& 1.9500& .0398& .1645\\
\hline
$\mu$ &3& 2.9998& .1062& .2581\\
\hline
\end{tabular}
\begin{tabular}{|c|c|c|c|c|}
\hline
\multicolumn{5}{|c|}{Sample size $n = 500$, shape $\alpha = 1$} \\
\hline
 & Actual Value & Mean Estimated Value & Variance & RMSE \\
\hline
$\delta$ & 1& 1.0009& .0012& .0275\\
\hline
$\sigma$ & 2& 1.9956& .0039& .0505\\
\hline
$\mu$ & 3& 3.0010& .0094& .0770\\
\hline
\end{tabular}
\begin{tabular}{|c|c|c|c|c|}
\hline
\multicolumn{5}{|c|}{Sample size $n = 50$, shape $\alpha = 0.25$} \\
\hline
 & Actual Value & Mean Estimated Value & Variance & RMSE \\
\hline
$\delta$ & 0 & 0.0000 & 0.0000 & 0.0004 \\
\hline
$\sigma$ & 1 & 0.9759 & 0.0103 & 0.0840 \\
\hline
$\mu$ & 0 & -0.0006 & 0.0872 & 0.2319 \\
\hline
\end{tabular}
\begin{tabular}{|c|c|c|c|c|}
\hline
\multicolumn{5}{|c|}{Sample size $n = 500$, shape $\alpha = 0.25$} \\
\hline
 & Actual Value & Mean Estimated Value & Variance & RMSE \\
\hline
$\delta$ & 0 & 0.0000 & 0.0000 & 0.0000 \\
\hline
$\sigma$ & 1 & 0.9973 & 0.0010 & 0.0251 \\
\hline
$\mu$ & 0 & 0.0001 & 0.0081 & 0.0722 \\
\hline
\end{tabular}
\begin{tabular}{|c|c|c|c|c|}
\hline
\multicolumn{5}{|c|}{Sample size $n = 50$, shape $\alpha = 2$} \\
\hline
 & Actual Value & Mean Estimated Value & Variance & RMSE \\
\hline
$\delta$ & 0 & -0.0006 & 0.0488 & 0.1728 \\
\hline
$\sigma$ & 1 & 0.9739 & 0.0098 & 0.0823 \\
\hline
$\mu$ & 0 & 0.0031 & 0.0219 & 0.1168 \\
\hline
\end{tabular}
\begin{tabular}{|c|c|c|c|c|}
\hline
\multicolumn{5}{|c|}{Sample size $n = 500$, shape $\alpha = 2$} \\
\hline
 & Actual Value & Mean Estimated Value & Variance & RMSE \\
\hline
$\delta$ & 0 & 0.0004 & 0.0041 & 0.0509 \\
\hline
$\sigma$ & 1 & 0.9972 & 0.0010 & 0.0257 \\
\hline
$\mu$ & 0 & -0.0001 & 0.0021 & 0.0362 \\
\hline
\end{tabular}
\begin{tabular}{|c|c|c|c|c|}
\hline
\multicolumn{5}{|c|}{Sample size $n = 50$, shape $\alpha = 5$} \\
\hline
 & Actual Value & Mean Estimated Value & Variance & RMSE \\
\hline
$\delta$ & 0 & -0.0039 & 0.4507 & 0.5321 \\
\hline
$\sigma$ & 1 & 0.9743 & 0.0098 & 0.0822 \\
\hline
$\mu$ & 0 & 0.0004 & 0.0222 & 0.1179 \\
\hline
\end{tabular}
\begin{tabular}{|c|c|c|c|c|}
\hline
\multicolumn{5}{|c|}{Sample size $n = 500$, shape $\alpha = 5$} \\
\hline
 & Actual Value & Mean Estimated Value & Variance & RMSE \\
\hline
$\delta$ & 0 & 0.0032 & 0.0399 & 0.1599 \\
\hline
$\sigma$ & 1 & 0.9977 & 0.0010 & 0.0254 \\
\hline
$\mu$ & 0 & -0.0011 & 0.0020 & 0.0358 \\
\hline
\end{tabular}
\vspace{0.5cm}
\caption{Simulation study: Distribution of MLE for $\delta, \mu, \sigma$ for sample size $n = 50, 500$ and shape parameters $\alpha = 1, 0.25, 2, 5$ with $\beta=1$.}
\label{table:sim}
\end{table}

\section{Financial Modeling}
\label{bggl.financial.sec}

\subsection{Background} 
Consider a stock market index, say Standard \& Poor 500 (S\&P 500) weekly data. Let $S(t)$ be end-of-week closing data. The {\it price returns} are defined as log changes in price: $Y(t) = \ln S(t) - \ln S(t-1) = \ln\frac{S(t)}{S(t-1)}$. For some indices, we can observe its volatility index measured daily. Let $V(t)$ be the value of this index: weekly average of daily closing prices. It plays the role of the standard deviation of $Y(t)$. Measuring volatility attempts to capture the observation that this standard deviation is not constant over time. During crises, volatility is high, but during normal times, it is low. Eras of high volatility alternate with eras of low volatility. In other words, volatility has {\it persistence:} If it is high this week it is likely to stay high next week. If it is low this week it is likely to stay low next week. To model this, one can use autoregression of order 1. Since $V$ is always positive, it is natural to apply the logarithmic transform,
\begin{equation}
\label{eq:AR}
\ln V(t) = a + b\ln V(t-1) + \zeta(t),
\end{equation}
with IID $\zeta(t)$ with $\mathbb E[\zeta(t)] = 0$. Having interpreted $V(t)$ as the standard deviation of returns $Y(t)$, we can write $Y(t) = V(t)\xi(t)$ for IID $\xi(t)$. More generally, we can model
\begin{equation}
\label{eq:nreturns}
Y(t) = c + hV(t) + V(t)\eta(t),
\end{equation}
where $\eta(t)$ are IID with mean zero. Fitting~\eqref{eq:AR} and~\eqref{eq:nreturns} is the content of the article \cite{VIX}, for monthly data 1986--2024. In turns out that $\eta$ is well-explained by a normal distribution, but $\zeta$ is not. A more appropriate distribution for $\zeta$ is, in fact, variance-gamma (GAL). But even this fitting is challenging, see \cite[Section 3]{VIX}. Also, $\eta$ and $\zeta$ are correlated. 
One could possibly model $(\eta, \zeta)$ using bivariate GAL, since normal distribution is a limiting case of a one-dimensional GAL. This is left for future research. The model~\eqref{eq:AR} and~\eqref{eq:nreturns} looks similar to a classic Heston stochastic volatility (SV) model with two innovation sequences, as opposed to GARCH models with only one innovation sequence. For background on SV, see the monograph \cite{SV-book}. But the SV models consider volatility to be hidden and only returns to be observed. This is why in classic SV models volatility needs to be inferred using hidden Markov models or other complicated statistical techniques. In the model given by~\eqref{eq:AR} - \eqref{eq:nreturns} we observe both the returns and the volatility. Thus, we can estimate $a, b$ and find the distribution of $\zeta(t)$ separately from $c, h$ and $\eta(t)$. 

\begin{figure}
\centering
\subfloat[$X$ vs Gamma]{\includegraphics[width=7cm]{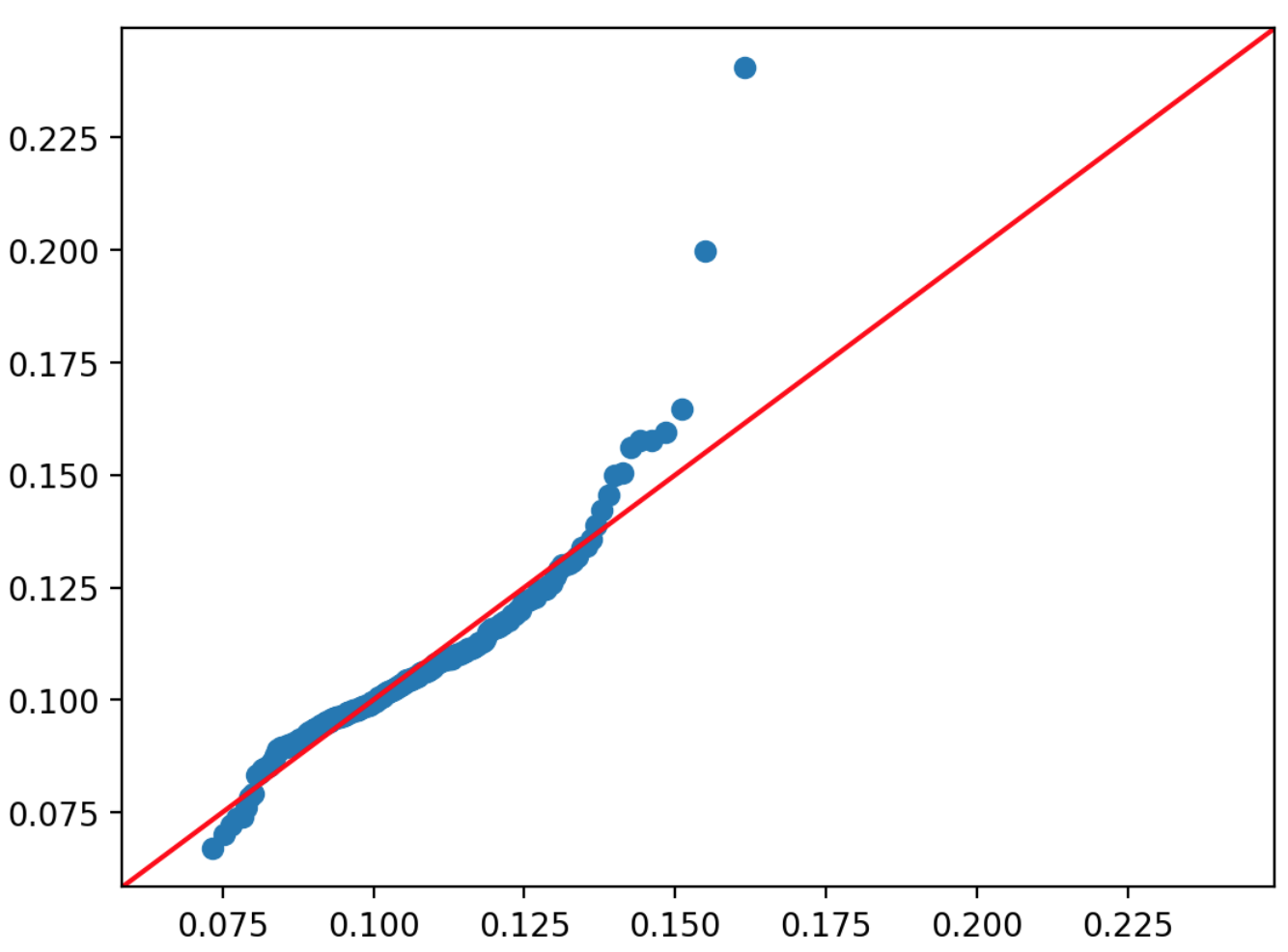}}
\subfloat[$Z$ vs Normal]{\includegraphics[width=7cm]{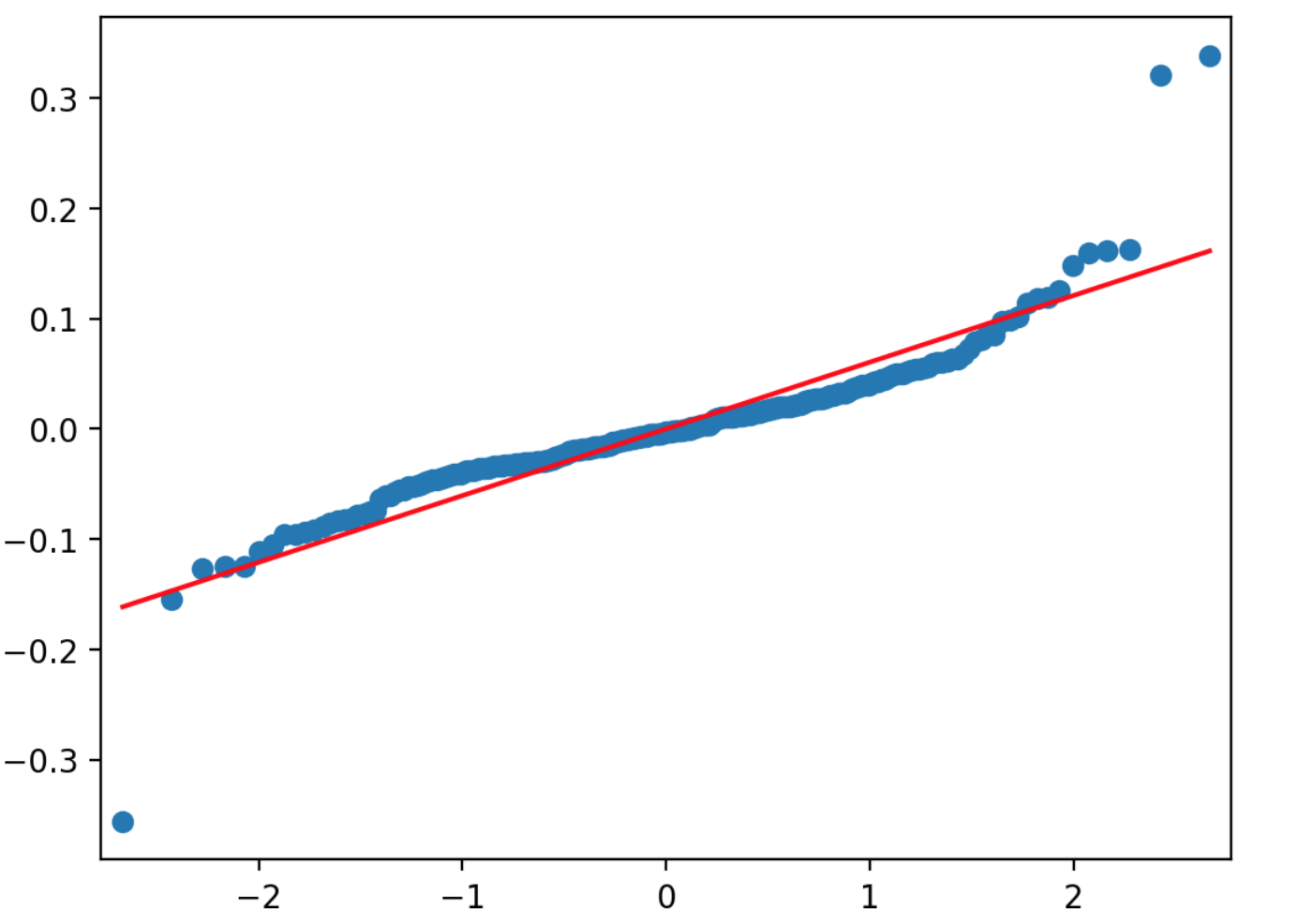}}
\vspace{0.5cm}
\caption{The quantile-quantile plots of $X$ and $Z$ from financial data of S\&P 500 returns and volatility versus their theoretical distributions.}
\label{fig:QQ}
\end{figure}

\subsection{Our contributions} 
Here, we take another approach. We insert stochastic volatility $V$ into returns data $Y$ itself. However, we cannot work directly with $(Y(t), V(t))$ since $V(t)$ is dependent upon $V(t-1)$. Let us take $\zeta(t)$ instead of $V(t)$. To make it into a Gamma random variable, we use the exponential function, since $\zeta(t)$ can be negative. Define 
\begin{equation}
\label{eq:X-vol}
X(t) = \exp(\zeta(t)) = e^{-a}\cdot V(t)\cdot V^{-b}(t-1).
\end{equation}
Let us model $(X(t), Y(t))$ as BGGL. As in \cite{VIX}, we can directly observe weekly returns $Y$ and weekly volatility $V$, and therefore compute $X$ from~\eqref{eq:X-vol}. 

We consider three indices for stock markets in the United States of America: 
\begin{itemize}
\item Standard \& Poor 500 (S\&P 500, a broad index of stocks for 500 large companies, capitalization-weighted);
\item Dow Jones Industrial Average (DJIA, an index 30 large stocks, traditional-oriented, price-weighted); 
\item NASDAQ 100 (another capitalization-weighted index, technology-oriented). 
\end{itemize}
Each of these three indices has its own volatility index: VIX for S\&P 500; VXD for DJIA; VXN for NASDAQ 100. We examined weekly data from September 1, 2019 to September 1, 2024, taken from Yahoo Finance. The results of maximum likelihood estimation, discussed  in Section 3, are shown in Table~\ref{table:finance-results}. 
\begin{table}
\begin{tabular}{|c|c|c|c|c|c|c|}
\hline
Stock Index & Volatility Index &  $\alpha$ & $\beta$ & $\delta$ & $\mu$ & $\sigma$ \\
\hline
S\&P 500 & VIX & 1.1399 & 8.3909 & 0.0219 & -0.0009 & 0.0048 \\
\hline
DJIA & VXD & 1.2639 & 12.864 & 0.01112 & -0.00029 & 0.00401 \\
\hline
NASDAQ 100 & VXN & 1.82288 & 7.81255 & 0.02777 & -0.00078 & 0.00502 \\
\hline
\end{tabular}
\vspace{0.5cm}
\caption{MLE for Stock Index Returns and Volatility}
\label{table:finance-results}
\end{table}
How good is the fit? There are three ways to judge. First, the fit of $X$ versus Gamma with fitted parameters $\alpha, \beta$. Second, the fit of $Y$ versus GAL with fitted parameters, as in~\eqref{eq:VG}. Third, the fit of $Z$: extracted regression residuals  (computed from $X$ and $Y$ and fitted parameters) versus Gaussian. 

We assess the quality of each fit graphically using quantile - quantile (QQ) plots. Overall, the fits appear satisfactory, with only a few deviations attributable to outlying observations. Figure~\ref{fig:QQ} illustrates the results for the S\&P 500 data; similar QQ plots are obtained for the DJIA and NASDAQ indices.

\subsection{Interpretation} 
We stress that conceptually this BGGV model is similar to~\eqref{eq:AR} and~\eqref{eq:nreturns}. It has returns on the left-hand side of an equation, and a function $X$ of volatility times a standard normal random variable on the right-hand side of this equation. This is similar to~\eqref{eq:nreturns}. The function $X$ of volatility is not exactly volatility at a given time. Rather, it depends on the volatility at times $t$ and $t-1$. More precisely, $X(t)$ captures how large is this week's volatility compared to what you would expect given last week's volatility, according to this autoregression~\eqref{eq:AR}. One can view this $X$ as unexpected or surprise volatility. We are able to extract IID terms from volatility to model them together with returns by a simple bivariate distribution instead of a time series model like \cite{VIX} or GARCH or SV. 

To conclude, here we incorporated stochastic volatility inside returns and modeled the latter using IID. In fact, we switched back from stochastic volatility models to modeling returns as IID and index levels as geometric random walk (but with non-normal increments).

\section{Appendix} 
\label{bggl.appen.sec}
\noindent In this section we present longer technical proofs and auxiliary results. 

\subsection{A technical lemma} Fix numbers $x_1, \ldots, x_n > 0$ and $y_1, \ldots, y_n$ with $\overline{x} = (x_1 + \cdots + x_n)/n$ and $\overline{y} = (y_1 + \cdots + y_n)/n$. 
\begin{lemma}
\label{fun.w.lem}
We have:
\begin{equation}
\label{bgal.fun.w}
w(\delta) =  \frac{(\overline{y} - \delta)^2}{\overline{x}} - 
\frac{1}{n}  \sum_{i=1}^n \frac{(y_i-\delta)^2}{x_i} \le 0,\quad \delta \ge 0.
\end{equation}
\end{lemma}
\begin{proof}
Observe that $\overline{y}_n - \delta = \frac{1}{n}  \sum_{i=1}^n (y_i-\delta)$, 
so  the inequality $w(\delta)\leq 0$ is equivalent to 
\begin{equation}
\label{kat2}
\Bigl[ \sum_{i=1}^n (y_i-\delta)\Bigr]^2 \leq  \sum_{i=1}^n x_i  \sum_{i=1}^n \frac{(y_i-\delta)^2}{x_i}
\end{equation}
However, this follows from Cauchy-Schwarz inequality, $(\sum_{i=1}^n a_ib_i)^2 \leq  \sum_{i=1}^n a_i^2  \sum_{i=1}^n b_i^2$,  if we identify $a_i=\sqrt{x_i}$ and $b_i=(y_i-\delta)/\sqrt{x_i}$. 
\end{proof}

\subsection{Proof of Proposition \ref{mix.ent.prop}}
\noindent Without loss of generality, we shall assume that $X$ has continuous distribution with the PDF $f_X(\cdot)$. By the mixture representation (\ref{eq:VG}), the joint PDF of $(X,Y)$ is given by 
\begin{equation}
\label{eq:conditional-density}
f(x,y) = f(y \mid x)f_X(x),
\end{equation}
where the conditional density of $Y$ given $X$ is normal:
\begin{equation}
\label{kruk1}
f(y \mid x) =  \frac{1}{\sqrt{2\pi x}}\frac{1}{\sigma} \exp\left[- \frac{(y-\delta-\mu x)^2}{2\sigma^2 x} \right], \,\,\, x > 0, \, y \in \mathbb R.
\end{equation}
This is the PDF of the conditional normal distribution in (\ref{eq:mixture}). By definition of the Shannon entropy, 
\begin{equation}
\label{eq:Shannon-formula}
H((X,Y))  =  - \int_0^\infty \int_{-\infty}^\infty \log[f(x,y)]f(x,y)\, \mathrm{d}y\, \mathrm{d}x.
\end{equation}
Applying~\eqref{eq:conditional-density} to~\eqref{eq:Shannon-formula}, we get:
\begin{align*}
H((X, Y)) &= - \int_0^\infty \log[f_X(x)]f_X(x) \int_{-\infty}^\infty f(y|x)\,\mathrm{d}y\, \mathrm{d}x  \\ & - \int_0^\infty f_X(x) \int_{-\infty}^\infty \log[f(y|x)] f(y|x)\,\mathrm{d}x.
\end{align*}
Apply Shannon entropy formula $H(Y\mid X = x)$ to conditional distribution of $Y$ given $X$:
\begin{align*}
H((X, Y)) &= - \int_0^\infty \log[f_X(x)]f_X(x) \int_{-\infty}^\infty f(y|x)\,\mathrm{d}y\, \mathrm{d}x \\ &  + \int_0^\infty f_X(x) H(Y\mid X=x)\, \mathrm{d}x.
\end{align*}
Applying~\eqref{eq:conditional-density} and using the definition of Shannon entropy $H(X)$ for only $X$, we obtain 
\begin{align*}
H((X, Y)) &= H(X) +  \int_0^\infty \frac{1}{2} f_X(x) \,\log [2\pi e \sigma^2 x] \,\mathrm{d}x\\    & = H(X) +  \frac{1}{2} \log [2\pi e \sigma^2] + \int_0^\infty \frac{1}{2} f_X(x)\,\log x\, \mathrm{d}x.
\end{align*}
Here, $H(Y|X=x) = \frac{1}{2} \log [2\pi e \sigma^2 x]$ is the Shannon entropy of the (normal!) conditional distribution of $Y|X=x$, whose PDF is given in (\ref{kruk1}). This proves the result. 

\subsection{Proof of Theorem  \ref{bggl.mle.theo}}
\noindent First, observe that for any particular values of $\delta$ and $\upsilon$, the maximum value of $g(\delta,\mu, \upsilon)$ is obtained by that value of $\mu\in \mathbb R$ that maximizes the function 
\begin{equation}
\label{bgal.fun.v}
s(\mu) =  \frac{\mu}{\upsilon}\frac{1}{n} \sum_{i=1}^n (Y_i-\delta) - \frac{\mu^2}{2n\upsilon}\frac{1}{n}  \sum_{i=1}^n X_i.
\end{equation}
This is a simple quadratic function of $\mu$ with a negative leading term, and it is maximized by a unique value
\begin{equation}
\label{bgal.mle.mu}
\mu(\delta) = \frac{\overline{Y} - \delta}{\overline{X}}.
\end{equation}
Here, the $\overline{X}$ and $\overline{Y}$ are the sample means of the $\{X_i\}$ and $\{Y_i\}$, respectively. This leaves us with the problem of maximizing the function $r(\delta, \upsilon) = g(\delta, \mu(\delta), \upsilon)$, which, after some routine algebra, can be shown to be 
\begin{equation}
\label{bgal.fun.r}
r(\delta, \upsilon) =  - \frac{1}{2}\log \upsilon + \frac{1}{2\upsilon} \left\{\frac{(\overline{Y} - \delta)^2}{\overline{X}} - 
\frac{1}{n}  \sum_{i=1}^n \frac{(Y_i-\delta)^2}{X_i} \right\}, \quad  \delta\in \mathbb R, \upsilon \ge 0. 
\end{equation}
At this stage, we again consider $\upsilon>0$ to be fixed, and focus on the function of $\delta$ given by the curly bracket in \eqref{bgal.fun.r}, 
\begin{equation}
\label{eq:bgal.fun.w}
w(\delta) =  \frac{(\overline{Y} - \delta)^2}{\overline{X}} - 
\frac{1}{n}  \sum_{i=1}^n \frac{(Y_i-\delta)^2}{X_i}, \quad \delta\in \mathbb R. 
\end{equation}
Clearly, this is a quadratic function $\delta$, $w(\delta) = a\delta^2 + b\delta +c$, where 
\begin{equation}
\label{w.quad.fun}
a = \frac{1}{\overline{X}} - \frac{1}{n}\sum_{i=1}^n \frac{1}{X_i}, \,\,\, b = 2 \left( \frac{1}{n}\sum_{i=1}^n \frac{Y_i}{X_i} - \frac{\overline{Y}}{\overline{X}}\right), \,\,\, c = \frac{\overline{Y}^2}{\overline{X}} - \frac{1}{n}\sum_{i=1}^n \frac{Y_i^2}{X_i},
\end{equation}
unless the leading coefficient $a$ in (\ref{w.quad.fun}) is equal to zero. This situation occurs if all of the $\{X_i\}$ are the same. In the latter case, the coefficient $b$ in (\ref{w.quad.fun}) is also zero, and we have $w(\delta)=c$ for all $\delta\in \mathbb R$ with $c$ as in (\ref{w.quad.fun}). In this case the parameter $\delta$ is not uniquely estimable. However, this occurs with probability zero. 

In the rest of the proof, we shall assume that not all values of the $\{X_i\}$ are the same, so that $w(\delta)$ is truly a quadratic function. By Lemma~\ref{fun.w.lem}, we have $w(\delta)\leq 0$ for all $\delta\in \mathbb R$. Consequently, the function will attain its maximum value at the vertex, which can be calculated explicitly by setting the derivative of this function equal to zero. Routine algebra shows that this value of $\delta$ is equal to 
\begin{equation}
\label{bgal.mle.delta}
\hat{\delta}_n = \frac{\frac{1}{n}\sum_{i=1}^n\frac{Y_i}{X_i} - \frac{\overline{Y}}{\overline{X}}}{\frac{1}{n}\sum_{i=1}^n\frac{1}{X_i} - \frac{1}{\overline{X}}},
\end{equation}
which is the MLE of $\delta$. Finally, we substitute that MLE into the function $w$ in \eqref{eq:bgal.fun.w}, and conclude that the largest value of the function $r$ in \eqref{bgal.fun.r} when $\upsilon>0$ is held fixed is equal to 
\begin{equation}
\label{bgal.fun.u}
u(\upsilon) = r(\hat{\delta}_n, \upsilon) = - \frac{1}{2}\log \upsilon + \frac{1}{2\upsilon}w(\hat{\delta}_n), \quad  \upsilon \ge 0. 
\end{equation}
To conclude the estimation problem, it remains to maximize the function $u(\upsilon)$ defined above with respect to $\upsilon \ge 0$. We have now two possibilities: (i) either $w(\hat{\delta}_n)=0$ or (ii) $w(\hat{\delta}_n)<0$. If we have $w(\hat{\delta}_n)=0$ in (\ref{bgal.fun.u}), then the function $u(\upsilon) = - \log \upsilon/2$ is maximized by $\hat{\upsilon}_n=0$. If we have $w(\hat{\delta}_n)<0$ in (\ref{bgal.fun.u}), then it is easy to see that the $u(\upsilon)$ admits a finite maximum value, which is attained at a unique 

\begin{equation}
\label{bgal.mle.sigma}
\hat{\upsilon}_n = -w(\hat{\delta}_n) = 
\frac{1}{n}  \sum_{i=1}^n \frac{(Y_i-\hat{\delta}_n)^2}{X_i} - \frac{(\overline{Y} -\hat{\delta}_n)^2}{\overline{X}},
\end{equation}
where $\hat{\delta}_n$ is given by (\ref{bgal.mle.delta}). Case (i) is actually included in Case (ii): Under Case (i), the MLE of $\upsilon$ is still given by (\ref{bgal.mle.sigma}). In summary, the three MLEs always exit and are unique, and are given by (\ref{eq:delta-mu}) -- (\ref{eq:s2}). This concludes the proof. 

\subsection{Properties of Bessel functions} The content of this subsection is taken from \cite[Chapter 10]{Handbook}, and we refer to formulas there (see also \cite[Appendix]{Survey}). We use the modified Bessel function $K_{\nu} : (0, \infty) \to (0, \infty)$ of the second kind. For each $\nu \in \mathbb R$, consider the second-order linear homogeneous differential equation, given in (10.25.1):
$$
z^2f''(z) - zf'(z) - (z^2 + \nu^2)f(z) = 0.
$$
It has a unique solution which satisfies the asymptotics $f(z) \sim (\pi/2z)^{1/2}e^{-z}$ as $z \to +\infty$. This solution is denoted as $f(z) \equiv K_{\nu}(z)$. The formula (10.32.10) gives us
\begin{equation}
\label{eq:library}
\frac{2^{\nu+1}K_{\nu}(z)}{z^{\nu}} = \int_0^{\infty}\exp\left(-u-\frac{z^2}{4u}\right)\,\mathrm{d}u.
\end{equation}
The formula (10.27.3) gives us
\begin{equation}
\label{pro10}
K_\lambda(u) = K_{-\lambda}(u), 
\end{equation}
The formulas (10.30.2) and (10.30.3): if $\lambda$ is fixed and $x\rightarrow 0^+$, then 
\begin{equation}
\label{pro6}
K_{\lambda}(x) \sim \frac{1}{2}\Gamma(\lambda)(x/2)^{-\lambda} \,\,\, (\lambda>0), \,\,\, K_{0}(x) \sim -\ln x. 
\end{equation}
Finally, the formula (10.29.4) gives us 
\begin{equation}
\label{pro9}
[x^\lambda K_{\lambda}(x)]' = -x^\lambda K_{\lambda-1}(x). 
\end{equation}
\begin{lemma}
Define the function 
\begin{equation}
\label{eq:F-alpha}
F_{\alpha}(x) = x^{\alpha}K_{\alpha}(x),\quad 0 < \alpha \le 1,\quad x > 0.
\end{equation}

As $x \to 0^+$, we have: $ F_{\alpha}(x) \to F_\alpha(0^+) = \Gamma(\alpha)2^{\alpha - 1}$.  Moreover, as $x \to 0^+$, we have:
\begin{align} 
\label{eq:limits}
\begin{split}
\ln F_{\alpha}(x) - \ln F_{\alpha}(0^+) &\sim \gamma(x) := 
\begin{cases}
-\frac{\Gamma(1 - \alpha)}{\Gamma(1 + \alpha)}2^{- 2\alpha}x^{2\alpha},\, \alpha \in (0, 1);\\
\frac12x^2\ln x,\, \alpha = 1.
\end{cases}
\end{split}
\end{align}
\label{lemma:Bessel}
\end{lemma}
\begin{proof} The first limit in Lemma~\ref{lemma:Bessel} follows immediately from~\eqref{pro6}. The limit in~\eqref{eq:limits} requires more work: Apply L'Hospital's rule and get that
\begin{align}
\label{eq:main}
\begin{split}
&\lim\limits_{x \to 0+}\frac{\ln F_{\alpha}(x) - \ln F_{\alpha}(0^+)}{\gamma(x)} = \lim\limits_{x \to 0+}\frac{(\ln F_{\alpha}(x))'}{\gamma'(x)} \\ & = \lim\limits_{x \to 0^+}\frac{F'_{\alpha}(x)}{F_{\alpha}(x)\gamma'(x)} = \frac1{F_{\alpha}(0+)}\lim\limits_{x \to 0^+}\frac{F'_{\alpha}(x)}{\gamma'(x)}.
\end{split}
\end{align}
Now,  let us consider first the case $\alpha = 1$. Then, by~\eqref{pro9}, $F'_1(x) = xK_0(x)$. Using~\eqref{pro6}, we get: $F'_1(x) \sim x\ln x$, and $\gamma'(x) = x\ln x + 0.5x \sim x\ln x$ as $x \to 0^+$. It follows that the right-hand side of~\eqref{eq:main} is equal to 1.  The case $\alpha \in (0, 1)$ is slightly harder. Applying~\eqref{pro10} and~\eqref{pro9}, we get: $F'_{\alpha}(x) = -x^{\alpha}K_{\alpha - 1}(x) = -x^{\alpha}K_{1 - \alpha}(x)$. It follows from~\eqref{pro6} that 
\begin{equation}
\label{eq:F'}
F'_{\alpha}(x) \sim -\frac12\Gamma(1 - \alpha)x^{2\alpha - 1}2^{1 - \alpha},\quad x \to 0^+.
\end{equation}
Next, from the definition of the function $\gamma$, we get:
\begin{equation}
\label{eq:gamma'}
\gamma'(x) = -\frac{\Gamma(1 - \alpha)}{\Gamma(1 + \alpha)}2^{- 2\alpha}\cdot 2\alpha x^{2\alpha - 1}.
\end{equation}
Combining~\eqref{eq:F'} and~\eqref{eq:gamma'}, we get that the right-hand side of~\eqref{eq:main} equals 1 as well.   
\end{proof}

\subsection{Proof of Theorem~\ref{thm:<1}} We only need to prove the convergence for $(\hat{\delta}_n, \hat{\mu}_n)$, since these are asymptotically independent of $(\hat{\alpha}_n, \hat{\beta}_n)$ (follows from Theorem~\ref{thm:unbiased}), and all these are asymptotically independent of $\hat{\upsilon}_n$. The asymptotic result for $\hat{\upsilon}_n$ follows from the fact that $\chi^2_{n-2}$ is the sum of squares of $n-2$ independent identically distributed standard normal random variables, and from the standard probability arguments. The results of Lemma~\ref{lemma:Bessel} imply the convergence, with IID standard normal $Z_1, Z_2 \sim \mathcal N(0, 1)$:
\begin{equation}
\label{eq:fast}
n^{1/2\alpha}(\hat{\delta}_n - \delta) \stackrel{d}{\rightarrow} \sigma  \beta^{-1/2}(\Gamma(\alpha+1))^{1/2\alpha}\cdot \xi_{\alpha}^{-1/2} Z_1,
\end{equation}
\begin{equation}
\label{eq:slow}
\sqrt{n}(\hat{\mu}_n - \mu) \stackrel{d}{\rightarrow} \mathcal \sigma \sqrt{\beta/\alpha} \cdot Z_2.
\end{equation}
To establish the joint convergence in (\ref{stef1}), we proceed as follows. As noted in Theorem~\ref{thm:unbiased}, the conditional distribution of $(\hat{\delta}_n, \hat{\mu}_n)$ (conditional on $X_1, \ldots, X_n$) is bivariate Gaussian:
$$
\begin{bmatrix}\hat{\delta} \\ \hat{\mu}\end{bmatrix} \sim \mathcal N_2\left(\begin{bmatrix}\delta \\ \mu\end{bmatrix}, \sigma^2 \begin{bmatrix} \sum_{i=1}^n X_i^{-1} & n\\ n & \sum_{i=1}^n X_i\end{bmatrix}^{-1}\right).
$$
Consequently, the conditional distribution of 
\begin{equation}
\label{stef2}
\begin{bmatrix} n^{\frac{1}{2\alpha}} & 0\\ 0 & n^{\frac{1}{2}}\end{bmatrix} \left( \begin{bmatrix}\hat{\delta}_n \\ \hat{\mu}_n\end{bmatrix} -  \begin{bmatrix}\delta \\ \mu \end{bmatrix}  \right)
\end{equation}
is also bivariate normal, with vector mean zero and covariance matrix given by 
\begin{equation}
\label{stef3}
\sigma^2 \begin{bmatrix} n^{\frac{1}{2\alpha}} & 0\\ 0 & n^{\frac{1}{2}}\end{bmatrix} \begin{bmatrix} \sum_{i=1}^n X_i^{-1} & n\\ n & \sum_{i=1}^n X_i\end{bmatrix}^{-1} \begin{bmatrix} n^{\frac{1}{2\alpha}} & 0\\ 0 & n^{\frac{1}{2}}\end{bmatrix}.
\end{equation}
Straightforward matrix multiplication shows that the covariance matrix in (\ref{stef3}) reduces to 
\begin{equation}
\label{stef4}
\sigma^2  \begin{bmatrix} n^{-1/\alpha }\sum_{i=1}^n X_i^{-1} & n ^{-\frac{1-\alpha}{2\alpha}}\\ n ^{-\frac{1-\alpha}{2\alpha}} & n^{-1} \sum_{i=1}^n X_i\end{bmatrix}^{-1}.
\end{equation}
The above discussion shows that we have the following stochastic representation:
\begin{equation}
\label{stef5}
\begin{bmatrix} n^{\frac{1}{2\alpha}} & 0\\ 0 & n^{\frac{1}{2}}\end{bmatrix} \left( \begin{bmatrix}\hat{\delta}_n \\ \hat{\mu}_n\end{bmatrix} -  \begin{bmatrix}\delta \\ \mu \end{bmatrix}  \right)
\stackrel{d}{=} \sigma \begin{bmatrix} n^{-1/\alpha }\sum_{i=1}^n X_i^{-1} & n ^{-\frac{1-\alpha}{2\alpha}}\\ n ^{-\frac{1-\alpha}{2\alpha}} & n^{-1} \sum_{i=1}^n X_i\end{bmatrix}^{-1/2} \begin{bmatrix}Z_1 \\ Z_2 \end{bmatrix}, 
\end{equation}
where the random matrix 
\begin{equation}
\label{stef6}
\begin{bmatrix} n^{-1/\alpha }\sum_{i=1}^N X_i^{-1} & n ^{-\frac{1-\alpha}{2\alpha}}\\ n ^{-\frac{1-\alpha}{2\alpha}} & n^{-1} \sum_{i=1}^n X_i\end{bmatrix}
\end{equation}
is independent of ${\bf Z} = (Z_1, Z_2) \sim N_2(\boldsymbol 0, {\bf I}_2)$. Thus, the convergence in (\ref{stef1}) will hold if we prove that 
\begin{equation}
\label{stef7}
\begin{bmatrix} n^{-1/\alpha }\sum_{i=1}^n X_i^{-1} & n ^{-\frac{1-\alpha}{2\alpha}}\\ n ^{-\frac{1-\alpha}{2\alpha}} & n^{-1} \sum_{i=1}^n X_i\end{bmatrix} 
\stackrel{d}{\rightarrow} 
\begin{bmatrix} \beta [\Gamma(\alpha+1)]^{-1/\alpha} \xi_\alpha & 0\\ 0 & \alpha/\beta \end{bmatrix}
\end{equation}
and apply the power $-1/2$ to the limiting matrix in~\eqref{stef7}. Since the off-diagonal elements of the matrix in (\ref{stef6}) are deterministic sequences converging to zero (as $0<\alpha<1$), the convergence in (\ref{stef7}) follows from the following lemma.
\begin{lemma}
\label{ope.con.lem}
If $\{X_i\}$ are IID random variables where $X_i\sim \mathrm{GAM}(\alpha, \beta)$, $i\in \mathbb N$, then 
\begin{equation}
\label{stef8}
\begin{bmatrix} n^{-\frac{1}{\alpha}} & 0\\ 0 & n^{-1}\end{bmatrix} \begin{bmatrix}\sum_{i=1}^n X_i ^{-1} \\  \sum_{i=1}^n X_i \end{bmatrix} 
\stackrel{d}{\rightarrow} 
\begin{bmatrix} \beta [\Gamma(\alpha+1)]^{-1/\alpha} \xi_\alpha \\   \alpha/\beta  \end{bmatrix}.
\end{equation}
\end{lemma}
\begin{proof} The convergence of the second component follows from Law of Large Numbers. From the Slutsky theorem, to establish~\eqref{stef8}, it is only sufficient to show
\begin{equation}
\label{eq:1d}
n^{-1/\alpha}\sum_{i=1}^n X_i ^{-1} \stackrel{d}{\to} \beta\left[\Gamma(\alpha + 1)\right]^{-1/\alpha}\xi_{\alpha}.
\end{equation}
In turn, to show~\eqref{eq:1d}, we shall show that $\phi_n(t) \rightarrow \phi(t)$ for $t \ge 0$, where $\phi_n(t)$ is the Laplace transform of the random sequence on the left-hand side of~\eqref{eq:1d} and 
\begin{equation}
\label{stef10}
\phi(t) = \mathbb E \exp\left[- t  \beta [\Gamma(\alpha+1)]^{-1/\alpha} \xi_\alpha\right] = \exp\left[ - \frac{\Gamma(1-\alpha)}{\Gamma(1+\alpha)} \beta^\alpha t ^\alpha \right] , \quad t \ge 0,
\end{equation}
is the Laplace transform of the random vector on the right-hand side of~\eqref{eq:1d}. First, note 
\begin{equation}
\label{stef11}
\phi_n(t) = \mathbb E \exp\left[-t n^{-\frac{1}{\alpha}}\sum_{i=1}^n X_i ^{-1} \right] = \left[\psi\left(n^{-\frac{1}{\alpha}}t\right)\right]^n,
\end{equation}
where $\psi(t) := \mathbb E \exp\left[- t /X\right]$ for $t \ge 0$ is the Laplace transform of the random variable $1/X$ with $X\sim \mathrm{GAM}(\alpha, \beta)$. Straightforward algebra shows that 
\begin{equation}
\label{stef13}
\psi(t) = \frac{\beta^\alpha}{\Gamma(\alpha)} \int_0^\infty x^{\alpha-1} e^{-\beta x - t/x}dx. 
\end{equation}
To evaluate this integral, we use the modified Bessel function of the second kind. We change the scale in~\eqref{stef13} inside the integral to make it look like the one in~\eqref{eq:library}: $u = \beta x$ and $\mathrm{d}u = \beta\,\mathrm{d}x$ and 
$$
\int_0^{\infty}x^{\alpha-1}e^{-\beta x - t/x}\,\mathrm{d}x = \int_0^{\infty}\left(\frac{u}{\beta}\right)^{\alpha - 1}e^{-u - t\beta/u}\,\frac{\mathrm{d}u}{\beta} = \beta^{-\alpha}\int_0^{\infty}u^{\alpha-1}e^{-u-t\beta/u}\,\mathrm{d}u.
$$
This matches the integral in~\eqref{eq:library} with $z^2/4 = t\beta \Leftrightarrow z = 2\sqrt{t\beta}$ and $\alpha = -\nu$. Then 
\begin{equation}
\label{eq:psi-t}
\psi(t) = \frac{\beta^{\alpha}}{\Gamma(\alpha)}\cdot \beta^{-\alpha}\int_0^{\infty}u^{\alpha-1}e^{-u-t\beta/u}\,\mathrm{d}u = \frac1{\Gamma(\alpha)}\int_0^{\infty}u^{\alpha-1}e^{-u-t\beta/u}\,\mathrm{d}u.
\end{equation}
Plug~\eqref{eq:library} into~\eqref{eq:psi-t} and get:
\begin{equation}
\label{eq:psi-final}
\psi(t) = \frac1{\Gamma(\alpha)}\frac{2^{-\alpha+1}K_{-\alpha}(2\sqrt{t\beta})}{(2\sqrt{t\beta})^{-\alpha}} = \frac{2(t\beta)^{\alpha/2}}{\Gamma(\alpha)}K_{-\alpha}(2\sqrt{t\beta}).
\end{equation}
By~\eqref{pro10}, we rewrite~\eqref{eq:psi-final} as 
\begin{equation}
\label{eq:psi-all}
\psi(t) = \frac{2(t\beta)^{\alpha/2}}{\Gamma(\alpha)}K_{\alpha}(2\sqrt{t\beta}).
\end{equation}
In turn, after plugging~\eqref{eq:F-alpha}, and using the first convergence in Lemma~\ref{lemma:Bessel}, the formula~\eqref{eq:psi-all} becomes
\begin{equation}
\label{eq:PSI}
\psi(t) = \frac{F_{\alpha}(2\sqrt{t\beta})}{\Gamma(\alpha)2^{\alpha-1}} \Rightarrow \ln\psi(t) = \ln F_{\alpha}(2\sqrt{t\beta}) - \ln F_{\alpha}(0+).
\end{equation}
Define the new variable 
\begin{equation}
\label{eq:h-n}
h_n := 2n^{-1/(2\alpha)}(\beta t)^{1/2} \Leftrightarrow n = \frac{2^{2\alpha}(t\beta)^{\alpha}}{h_n^{2\alpha}}.
\end{equation}
By~\eqref{eq:h-n},~\eqref{eq:PSI} and some algebra, we rewrite \eqref{stef11} as
\begin{equation}
\label{eq:psi-n}
\ln \phi_n(t) = n\ln \psi\left(n^{-1/\alpha}t\right) = \frac{\ln(F_{\alpha}(h_n)) - \ln F_{\alpha}(0+))}{h_n^{2\alpha}}\cdot 2^{2\alpha}t^{\alpha}\beta^{\alpha}.
\end{equation}
Now, $h_n \to 0+$ as $n \to \infty$. Applying Lemma~\ref{lemma:Bessel}, we have:
\begin{equation}
\label{eq:limit-v}
\ln\phi_n(t) \sim -\frac{\Gamma(1 - \alpha)}{\Gamma(1 + \alpha)}2^{- 2\alpha}\cdot 2^{2\alpha}t^{\alpha}\beta^{\alpha} = -\frac{\Gamma(1 - \alpha)}{\Gamma(1 + \alpha)}\beta^{\alpha}t^{\alpha}.
\end{equation}
Plugging this into~\eqref{eq:limit-v}, we get ~\eqref{stef10}. This concludes the proof. 
\end{proof}
\subsection{Proof of Theorem~\ref{thm:=1}} 
We closely follow the proof of Theorem~\ref{thm:<1}. The same remark at the beginning of the proof of Theorem~\ref{thm:=1} holds here. The two other equations~\eqref{stef2} and~\eqref{stef3} get modified as follows: The conditional distribution given $X_1, \ldots, X_n$  of 
$$
\begin{bmatrix}  (n\ln n)^{\frac12} & 0\\ 0 & n^{\frac{1}{2}}\end{bmatrix} \left( \begin{bmatrix}\hat{\delta}_n \\ \hat{\mu}_n\end{bmatrix} -  \begin{bmatrix}\delta \\ \mu \end{bmatrix}  \right)
$$
is also bivariate normal, with vector mean zero and covariance matrix given by 
$$
\sigma^2 \begin{bmatrix}  (n\ln n)^{\frac12} & 0\\ 0 & n^{\frac{1}{2}}\end{bmatrix} \begin{bmatrix} \sum_{i=1}^n X_i^{-1} & n\\ n & \sum_{i=1}^n X_i\end{bmatrix}^{-1} \begin{bmatrix}  (n\ln n)^{\frac12} & 0\\ 0 & n^{\frac{1}{2}}\end{bmatrix}.
$$
We can compute this covariance matrix using simple matrix arithmetic as
$$
\sigma^2\begin{bmatrix} 
(n\ln n)^{-1}\sum_{i=1}^nX_i^{-1} & (\ln n)^{-\frac12} \\ 
(\ln n)^{-\frac12} & n^{-1}\sum_{i=1}^nX_i
\end{bmatrix}.
$$
Thus it is sufficient to prove the following lemma. 
\begin{lemma}
\label{ope.con.lem-alpha1}
If $\{X_i\}$ are IID random variables where $X_i\sim \mathrm{GAM}(\alpha = 1, \beta)$, $i\in \mathbb N$, then as $n \to \infty$, we have convergence in distribution:
\begin{equation}
\label{stef8:a1}
\begin{bmatrix} (n\ln n)^{-1} & 0\\ 0 & n^{-1}\end{bmatrix} \begin{bmatrix}\sum_{i=1}^n X_i ^{-1} \\  \sum_{i=1}^n X_i \end{bmatrix} 
\stackrel{d}{\rightarrow} 
 \begin{bmatrix} \beta \\   \alpha/\beta  \end{bmatrix}.
\end{equation}
\end{lemma}
\begin{proof} We closely follow the proof of Lemma~\ref{ope.con.lem}. We need only to show
\begin{equation}
\label{eq:conv-alpha-1}
\frac1{n\ln n}\sum\limits_{k=1}^nX_i^{-1} \stackrel{d}{\to} \beta.
\end{equation}
To show~\eqref{eq:conv-alpha-1}, we prove convergence of Laplace transforms. Take $\alpha = 1$ in~\eqref{eq:PSI}. We get $F_1(0+) = 1$. And replace $n^{1/\alpha}$ with $n\ln n$ in~\eqref{stef11}. Then the Laplace transform $\phi_n$ of the left-hand side in~\eqref{eq:conv-alpha-1} has logarithm
\begin{equation}
\label{eq:psi-n-1}
\ln \phi_n(t) = n\ln\psi\left(t/n\ln n\right) = n\ln F_1\left(2\sqrt{t\beta/n\ln n}\right).
\end{equation}
Now, apply Lemma~\ref{lemma:Bessel} for the case $\alpha = 1$: The right-hand side of~\eqref{eq:psi-n-1} is equivalent asymptotically to 
$$
n\cdot \frac12\left(2\sqrt{\frac{t\beta}{n\ln n}}\right)^2\cdot \ln\left(2\sqrt{\frac{t\beta}{n\ln n}}\right) = n\cdot 2\frac{t\beta}{n\ln n}\cdot \left(-\frac12\ln n\right) = -t\beta. 
$$
This proves that $\phi_n(t) \to e^{-t\beta}$ as $n \to \infty$, and the right-hand side is the Laplace transform of the constant random variable $\beta$. This completes the proof of~\eqref{eq:conv-alpha-1}.
\end{proof}
\subsection{The MLE when all $X_i$ have the same value $X$}
Indeed, in this case the function $g(\delta, \mu, \upsilon)$ in~\eqref{eq:bgal.ll.fun} simplifies to 
\begin{equation}
\label{fun.g.spe}
g(\delta, \mu, \upsilon) = -\frac{1}{2} \log \upsilon +\frac{\mu}{\upsilon} (\overline{Y}-\delta) - \frac{1}{2\upsilon X n} \sum_{i=1}^n (Y_i-\overline{Y})^2 - \frac{1}{2\upsilon X}  (\overline{Y}-\delta)^2 - \frac{\mu^2}{2\upsilon} X. 
\end{equation}
Elementary analysis of the above function shows that it attains its maximum value of $g(\hat{\delta}_n, \hat{\mu}_n, \hat{\upsilon}_n) = - (\log \hat{\upsilon}_n +1)/2$ 
where $\hat{\delta}_n$ is any real number, $\hat{\mu}_n$ is as in (\ref{eq:delta-mu}), and 
\begin{equation}
\label{fun.g.spe.v_n}
\hat{\upsilon}_n =  \frac{1}{Xn} \sum_{i=1}^n (Y_i-\overline{Y})^2.
\end{equation}
Thus,  in this case the parameter $\delta$ is not uniquely estimable via MLE, while estimators of $\mu$ and $\upsilon$ are the same as those specified by Theorem \ref{bggl.mle.theo}, since the value of $\hat{\upsilon}_n$ in~\eqref{fun.g.spe.v_n} is exactly the same as that in~\eqref{eq:s2} assuming that all the values of $\{X_i\}$ are equal to the same $X>0$. In particular, this case arises when the sample size is $n=1$, where we additionally have $\hat{\upsilon}_n = 0$ and $g(\hat{\delta}_n, \hat{\mu}_n, \hat{\upsilon}_n) = \infty$.  Then the MLEs of $\delta$, $\mu$ and $\upsilon$ are not unique. Indeed, in this case by taking an arbitrary $\delta \neq Y$ and setting $\mu=(Y-\delta)/X$, the function $g(\delta, \mu,\upsilon)$ reduces to $-\log \upsilon/2$, with the latter approaching infinity as $\upsilon \rightarrow 0^+$.  

In addition, for $n=2$ we have $\sup_{\delta\in \mathbb R} w(\delta) = w(\hat{\delta}_n)=0$ 
, as can be verified directly. Thus, with a sample of two, the MLE of $\sigma$ is always equal to zero, pointing towards a degenerate distribution. On the other hand, when $n\geq 3$, then $\sup_{\delta\in \mathbb R} w(\delta) = w(\hat{\delta}_n)<0$ and consequently 
$\hat{\sigma}_n>0$, unless either all the $Y_i$ are equal or we have the relation 
\begin{equation}
\label{kat3}
\hat{\delta}_n = \frac{Y_iX_j - X_iY_j}{X_j-X_i}
\end{equation}
for all $i\neq j$, as can be deduced from general conditions for the equality in Cauchy-Schwartz inequality. However, this occurs with probability zero, and thus is not of a concern.

\end{document}